\documentclass[sigconf,balance=false,prologue,table,xcdraw]{acmart}
\usepackage{popets}

\setcopyright{popets}
\copyrightyear{YYYY}

\acmYear{YYYY}
\acmVolume{YYYY}
\acmNumber{X}
\acmDOI{XXXXXXX.XXXXXXX}
\acmISBN{}
\acmConference{Proceedings on Privacy Enhancing Technologies}
\usepackage[capitalise]{cleveref}

\usepackage{algorithm}
\usepackage{algpseudocode}

\usepackage{pgfplots}
\pgfplotsset{compat=1.17}  
\usepackage{subcaption}

\theoremstyle{definition}

\newtheorem{theorem}{Theorem}

\usepackage{booktabs}

\begin{document}

\title[Locality-Sensitive Hashing Does Not Guarantee Privacy!]{Locality-Sensitive Hashing Does Not Guarantee Privacy!\\Attacks on Google's FLoC and the MinHash Hierarchy System}



\author{Florian Turati}
\affiliation{%
  \institution{ETH Zurich}
  \city{Zurich}
  \country{Switzerland}}
\email{florian.turati@inf.ethz.ch}

\author{Carlos Cotrini}
\affiliation{%
  \institution{ETH Zurich}
  \city{Zurich}
  \country{Switzerland}}
\email{ccarlos@inf.ethz.ch}

\author{Karel Kubicek}
\orcid{0000-0002-7419-2784}
\affiliation{%
  \institution{ETH Zurich}
  \city{Zurich}
  \country{Switzerland}}
\email{karel.kubicek@inf.ethz.ch}

\author{David Basin}
\orcid{0000-0003-2952-939X}
\affiliation{%
  \institution{ETH Zurich}
  \city{Zurich}
  \country{Switzerland}}
\email{basin@inf.ethz.ch}

\renewcommand{\shortauthors}{Turati et al.}

\begin{abstract}
Recently proposed systems aim at achieving privacy using locality-sensitive hashing. We show how these approaches fail by presenting attacks against two such systems: Google's FLoC proposal for privacy-preserving targeted advertising and the MinHash Hierarchy, a system for processing mobile users' traffic behavior in a privacy-preserving way. Our attacks refute the pre-image resistance, anonymity, and privacy guarantees claimed for these systems.

In the case of FLoC, we show how to deanonymize users using Sybil attacks and to reconstruct 10\% or more of the browsing history for 30\% of its users using Generative Adversarial Networks. We achieve this only analyzing the hashes used by FLoC. For MinHash, we precisely identify the movement of a subset of individuals and, on average, we can limit users' movement to just 10\% of the possible geographic area, again using just the hashes. In addition, we refute their differential privacy claims.
\end{abstract}

\keywords{LSH, FLoC, MinHash, SimHash, Privacy}

\maketitle

\section{Introduction}

Locality-sensitive hashing (LSH)~\cite{rajaraman2011mining} is a group of hash functions that map, with high probability, similar objects to the same hash. Comparing hashes instead of entire objects then results in an efficient procedure that has been used, for example, for plagiarism detection~\cite{stein2007principles}, detecting duplicate websites or images~\cite{manku2007detecting, jing2008visualrank}, dimensionality reduction~\cite{brinza2010rapid}, and clustering~\cite{haveliwala2000scalable}.

Recent works~\cite{minhashprivacy, floc-whitepaper, Malekzadeh_2018, apple-csam} have used LSH to process sensitive data, where it is assumed that the hashes can be made public without compromising the users' privacy. For example, Google proposed FLoC~\cite{floc-whitepaper}, a method for private targeted advertising. It uses LSH to map browsing histories to hashes such that users with similar browsing histories are likely to have the same hash. The hashes are then grouped into \emph{cohorts}. The idea is that each cohort contains users with similar browsing histories. The advertiser then learns only each cohort's identifier rather than each user's browsing history.\looseness=-1

A second example is Apple CSAM~\cite{apple-csam}, designed to detect child sexual abuse material in iCloud photos while preserving user privacy. It uses LSH to map images to hashes such that similar images have the same hash. This allows Apple to detect if abusive images are on a device. The hashes are intended, however, to prevent Apple from learning anything not related to abusive images.

We illustrate how systems that attempt to provide privacy using LSH are flawed. In particular, LSH hashes leak information about the input, since they do not provide security properties like pre-image resistance. None of the referenced works, however, were concerned by the privacy implications of this information leakage or considered its seriousness. We therefore investigate the severity of the leakage by developing new attacks on two recent applications: FLoC~\cite{floc-whitepaper} and the MinHash Hierarchy system~\cite{minhashprivacy}. 

\textbf{FLoC} is a system for private targeted advertising, proposed by Google. It uses SimHash to cluster users so that users with similar browsing histories are likely to be in similar cohorts. It aims at providing $k$-anonymity~\cite{sweeney2002k} while keeping the cohorts useful for targeted advertisements. 

\textbf{MinHash Hierarchy} is a system for analyzing traffic trajectories. It computes statistics on trajectories of mobile devices for urban planning, while ensuring differential privacy for these devices. It works by having cell stations store hashes that represent subsets of all mobile devices passing by. 

In this paper, we present three kinds of critical attacks on FLoC, which we illustrate using the MovieLens dataset~\cite{Harper2015-cx}. First, we present a pre-image attack on SimHash using integer programming. We then design a Sybil attack~\cite{douceur2002sybil} that generates dozens of histories per second whose hash matches the target hash. We demonstrate how this attack breaks FLoC's $k$-anonymity, and hence we can identify individuals in cohorts. Furthermore, using Generative Adversarial Networks (GANs)~\cite{goodfellow2020generative, guo2017long}, we partially reconstruct plausible histories from just the target hash. With this attack, we show how to break FLoC's privacy claims and infer some of the websites visited by users. Specifically, we can reconstruct $10\%$ or more of the history of at least $30\%$ of the users. Although FLoC is no longer used by Google, we present these attacks here to highlight the privacy limitations of LSH-based systems.

In the context of the MinHash Hierarchy, we demonstrate using taxi trajectories in the city of Porto,\footnote{Porto is the second largest city of Portugal with 232 thousand citizens occupying 41 km$^2$. Our dataset covers roughly $80$ km$^2$, since it includes the surrounding urban area.} that we can decide for some individuals whether they followed a particular trajectory, violating their differential privacy guarantee. Our attack narrows down the taxi drivers' trajectories, on average, to around $10\%$ of the  city's area, which corresponds to a neighborhood of Porto.

In both of these proposals, our attacks show that if the hashes can measure the similarity of objects, then they also contain fingerprints of the object itself. The amount of data in this fingerprint is bounded by the hash size, and while larger hashes provide more utility, they also contain more sensitive information. We are the first to evaluate this information leakage by attacking proposals of significant importance and measuring the information that attackers gain. 
Implementations of our attacks and more information are available at \url{https://karelkubicek.github.io/post/floc}.
Although some countermeasures have been proposed, we discuss in Section~\ref{sec:mitigations} how they fail to prevent our attacks.

Overall, \textbf{our contributions} are the following. First, we present a \textbf{pre-image attack} on SimHash using integer programming. Second, we implement practical \textbf{Sybil attacks} on FLoC by finding pre-images of a target SimHash. We show how this breaks the $k$-anonymity promised for FLoC by isolating real users in a cohort. Third, using GANs, we implement a \textbf{privacy attack} that reconstructs more than $10\%$ of the browsing history of at least $30\%$ of users based only on the FLoC hash. We also show how to amplify this attack to increase the size of the reconstructed history by exploiting the changes in users' SimHash. Finally, for the MinHash Hierarchy system, we present \textbf{privacy and pre-image attacks} that identify a subset of the individuals that visited a given checkpoint. We also show that we can track users to narrow down their trajectory to an average of $10\%$ of the city area, which corresponds to a local neighborhood.

\section{Federated Learning of Cohorts}

In this section, we give some preliminaries on locality-sensitive hashing (\cref{sec:background-simhash}). Afterwards, we present SimHash, a class of LSH, and FLoC, a proposal for privacy-preserving targeted advertisements (\cref{sec:simhash-application}). 

\subsection{SimHash} \label{sec:background-simhash}

Locality-sensitive hashing (\emph{LSH}) is a class of hash functions mapping similar inputs to similar outputs~\cite{rajaraman2011mining}. These hash functions are usually neither collision nor pre-image resistant like cryptographic hash functions.

In this section, we explain locality-sensitive hashing (\emph{LSH}) and SimHash, a popular instance proposed by Charikar~\cite{charikar2002similarity} and used by FLoC for privacy-preserving targeted advertising. In particular, Google researchers used SimHash to detect near duplicate websites with the search crawler Googlebot and more recently to measure the similarity between two browsing histories in FLoC. We explain next how SimHash works in the context of browsing histories.

We describe how to compute the SimHash of length $\ell$ of a browsing history $D$, which we represent as a finite set of domains. First, we produce for each domain $d \in D$ a \emph{fingerprint vector}, which is a vector $\eta_{d} \in \mathbb{R}^\ell$ sampled from the standard multivariate Gaussian in $\mathbb{R}^\ell$ using a pseudo-random generator that takes $d$ as the seed. Then we compute $y^{(D)} = \sum_{d \in D} \eta_d$ and the SimHash is $z^{(D)} = \text{sgn}\left(y^{(D)}\right)$, where $\text{sgn}$ applies the elementwise sign function to each entry of $y^{(D)}$. Note that the SimHash $z^{(D)} \in \{0, 1\}^n$ is a binary vector.

We now give an intuition of why the SimHash is locally sensitive. Suppose that $D$ and $D'$ are two browsing histories of the same size that differ in only one element. Then the sets of fingerprint vectors for $D$ and $D'$ differ in at most one vector. As a result, the sum $y^{(D)}$ of fingerprint vectors in $D$ is probably similar to the sum $y^{(D')}$. Therefore, $z^{(D)}$ and $z^{(D')}$ are probably the same. Note that the greater the number of different elements that $D$ and $D'$ have, the more unlikely it is that $z^{(D)} = z^{(D')}$.

We illustrate this computation for $\ell=5$ in \cref{fig:simhash-comp-ex}.
We have two browsing histories with three domains and a 5-bit target SimHash. The two browsing histories only differ in one domain and the resulting SimHash values only differ in one bit. Note how a slight change in the input changed only one bit of the resulting SimHash.


\begin{table}[]
\caption{Example of a SimHash computation.} 
\label{fig:simhash-comp-ex}
\resizebox{\columnwidth}{!}{%
\begin{tabular}{@{}rrrrrr rrrrrr@{}}
\toprule
\textbf{Hist. 1} & \multicolumn{5}{c}{\textbf{Domain fingerprints $\eta_{d}$}} & \textbf{Hist. 2} & \multicolumn{5}{c}{\textbf{Domain fingerprints $\eta_{d}$}} \\ \midrule
google & \cellcolor[HTML]{3D85C6}2.03 & \cellcolor[HTML]{EEF5FA}0.18 & \cellcolor[HTML]{BFD7EC}0.67 & \cellcolor[HTML]{C3DAEE}0.62 & \cellcolor[HTML]{E88E8E}-0.88 & google & \cellcolor[HTML]{3D85C6}2.03 & \cellcolor[HTML]{EEF5FA}0.18 & \cellcolor[HTML]{BFD7EC}0.67 & \cellcolor[HTML]{C3DAEE}0.62 & \cellcolor[HTML]{E88E8E}-0.88 \\
youtube & \cellcolor[HTML]{D83E3E}-1.51 & \cellcolor[HTML]{D11A1A}-1.79 & \cellcolor[HTML]{F8DDDD}-0.26 & \cellcolor[HTML]{B6D1EA}0.76 & \cellcolor[HTML]{94BCE0}1.11 & youtube & \cellcolor[HTML]{D83E3E}-1.51 & \cellcolor[HTML]{D11A1A}-1.79 & \cellcolor[HTML]{F8DDDD}-0.26 & \cellcolor[HTML]{B6D1EA}0.76 & \cellcolor[HTML]{94BCE0}1.11 \\
facebook & \cellcolor[HTML]{F9FBFE}0.07 & \cellcolor[HTML]{FEFBFB}-0.03 & \cellcolor[HTML]{D73939}-1.55 & \cellcolor[HTML]{EFAFAF}-0.62 & \cellcolor[HTML]{639DD2}1.61 & netflix & \cellcolor[HTML]{D3E3F2}0.46 & \cellcolor[HTML]{BFD7EC}0.67 & \cellcolor[HTML]{ECF3FA}0.20 & \cellcolor[HTML]{DF6060}-1.24 & \cellcolor[HTML]{FDFEFF}0.03 \\ \midrule
\textbf{sum:} & \cellcolor[HTML]{C6DCEF}0.59 & \cellcolor[HTML]{D52D2D}-1.64 & \cellcolor[HTML]{E16D6D}-1.14 & \cellcolor[HTML]{B6D1EA}0.76 & \cellcolor[HTML]{4D8FCB}1.84 & \textbf{sum:} & \cellcolor[HTML]{A0C4E4}0.98 & \cellcolor[HTML]{E78787}-0.94 & \cellcolor[HTML]{C4DAEE}0.61 & \cellcolor[HTML]{F2F7FC}0.14 & \cellcolor[HTML]{E6F0F8}0.26 \\ \midrule
\textbf{sign:} & \cellcolor[HTML]{9FC5E8}1 & \cellcolor[HTML]{EA9999}0 & \cellcolor[HTML]{EA9999}\textbf{0} & \cellcolor[HTML]{9FC5E8}1 & \cellcolor[HTML]{9FC5E8}1 & \textbf{sign:} & \cellcolor[HTML]{9FC5E8}1 & \cellcolor[HTML]{EA9999}0 & \cellcolor[HTML]{9FC5E8}\textbf{1} & \cellcolor[HTML]{9FC5E8}1 & \cellcolor[HTML]{9FC5E8}1 \\ \bottomrule
\end{tabular}
}
\end{table}

\subsection{Application of SimHash to Privacy}\label{sec:simhash-application}

In this section, we present \emph{Federated Learning of Cohorts} (FLoC)~\cite{floc-whitepaper}.
FLoC is a proposal from Google researchers to partially replace third-party cookies and perform privacy-preserving targeted advertisements. The idea is that users are grouped into cohorts so that users with similar browsing histories are assigned to the same cohort. Each cohort then gets an identifier. Instead of revealing personal browsing histories to advertisers, only the cohorts' identifiers are revealed.

\subsubsection{Clustering of SimHashes}\label{sec:clustering-simhash}

The FLoC proposal states that a SimHash is computed in the client's browser and serves as a history fingerprint, and only the hash is shared with a centralized clustering server. This server assigns a user the cohort identifier, where a cohort is a cluster of users with similar SimHashes. We illustrate the clustering procedure in \cref{fig:floc-clustering-ex} and describe it below.

\begin{figure} 
    \centering
    \includegraphics[width=\linewidth]{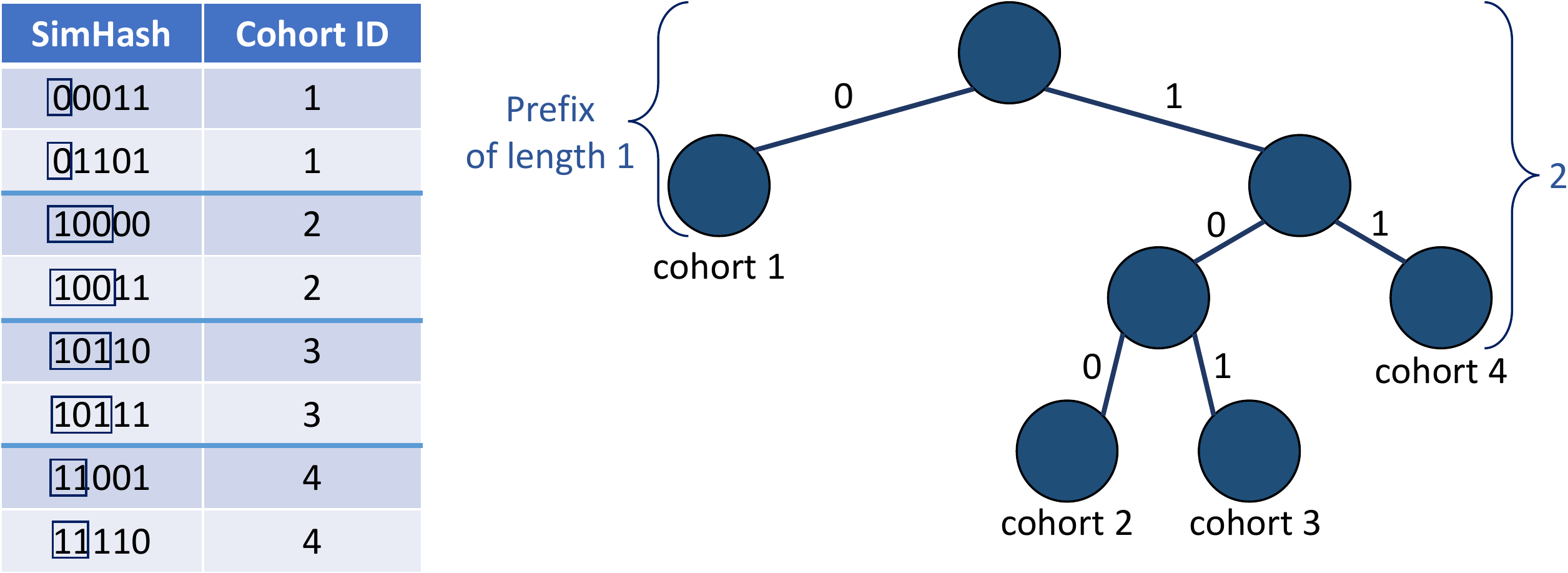} 
    \caption{Example of a clustering with FLoC}
    \label{fig:floc-clustering-ex}
\end{figure}

Let $\mathcal{D}$ be the set of browsing histories of a given set of users. For a bitstring $\sigma \in \{0, 1\}^*$, let $C_\sigma = \left\{D \in \mathcal{D} : \sigma \prec z^{(D)}\right\}$; that is, $C_\sigma$ contains all users (i.e., browsing histories) whose SimHash has $\sigma$ as a prefix. We call $C_\sigma$ a \emph{cohort} and we say that a cohort is \emph{$k$-decomposable}, for $k \in \mathbb{N}$, if $\left|C_{\sigma0}\right| \geq k$ and $\left|C_{\sigma1}\right| \geq k$. The clustering procedure starts with a clustering $\mathcal{C} = \{C_\epsilon\}$, where $\epsilon$ is the empty bitstring; that is, there is only one cluster at the start containing all users. Then a value $k \in \mathbb{N}$ is fixed. As long as there is a $k$-decomposable cluster $C_\sigma \in \mathcal{C}$, the procedure replaces $C_\sigma$ with $C_{\sigma0}$ and $C_{\sigma1}$. The idea is that each cohort in $\mathcal{C}$ provides $k$-anonymity, while containing a set of users with similar browsing histories.

\cref{fig:floc-clustering-ex} illustrates the result of the clustering procedure on a set of eight users. The table given there shows the SimHash of the browsing history of each user and an identifier of the cohort to which they have been assigned to by the clustering procedure. Note how each cohort has $k=2$ users. The tree in the figure illustrates how the clustering procedure divided $2$-decomposable cohorts until reaching the clustering assignment depicted in the table.

\subsubsection{Origin Trial}

From March $30$ to July $13$ $2021$, Google tested FLoC in its Origin trial~\cite{floc-chromium-ot}. Users of Chrome version numbers $89$ - $91$ located in ten countries 
were eligible for the experiment. Only $0.5\%$ of these eligible users were involved in the Origin trial, and only websites that requested a FLoC ID were added to the history used for FLoC computation. $50$-bit SimHashes were computed on a domain history of one week. Out of the $50$ bits, only $13$ to $20$ bits were necessary to split the users into around $33\,000$ cohorts of at least $2000$ users.

Despite the very small sample of users, some advertisers were successful in identifying topics of interest for users in the cohorts. For example, we refer to Criteo's blog~\cite{criteo-flocot} for the evolution over time of a cohort with around $10\,000$ users. We summarize their world cloud representation of the most popular topics in \cref{tab:criteo-word-cloud-evolution-large-cohort}, keeping only the five main topics. Also note that the main topics would vary a lot more in the case where only a small number of users can be observed.

\begin{table}
\centering
\caption{Criteo's example of the evolution of the most frequent topics browsed by a large cohort ($\approx 10\,000$ users)\label{tab:criteo-word-cloud-evolution-large-cohort}}
\resizebox{\columnwidth}{!}{%
\begin{tabular}{c c c}
 \toprule
 Week 0 & Week 3 & Week 5  \\ 
 \midrule
  Gaming & Gaming &  Tech. \& Computing   \\ 
  Tech. \& Computing & Tech. \& Computing & Gaming \\
  Books and Literature  & Education & News and Politics \\
  Education  & Shopping & Style \& Fashion\\
  Shopping & News and Politics & Healthy Living  \\
 \bottomrule
\end{tabular}
} 
\end{table}
CafeMedia, an ad management service, also analyzed the quality of cohorts for targeted advertising~\cite{cafemedia-flocot}. They formed groups of 1000 cohorts and computed the most frequent 10 keywords occurring in the browsing histories in those groups. \cref{tab:cafemedia-keywords} shows the top 10 keywords for $5$ groups. They could, for example, distinguish groups that are more interested in business and professional development and also groups that are more interested in leisure activities.

\begin{table}
\centering
\caption{CafeMedia's extracted interest keywords of the selected cohorts (see the complete table in~\cite{cafemedia-flocot})\label{tab:cafemedia-keywords}}
\resizebox{\columnwidth}{!}{%
\begin{tabular}{c c c c c c}
     \toprule
    Cohort IDs & \multicolumn{5}{c}{Keywords}  \\
     \midrule
     \textbf{0-1k} & music & support & grade & questions & season  \\ 
     \textbf{1k-2k} & dogs & guides & working & things & roast  \\ 
     \textbf{2k-3k} & writing & magic & vegetables & movies & slow  \\ 
     \textbf{3k-4k} & prime & high & rolls & magic & chili  \\
     \textbf{4k-5k} & weekly & world & disney & magic & sheets  \\
     \bottomrule
\end{tabular}
} 
\end{table}

\section{Attacks on FLoC}\label{sec:attack-floc}

In this section, we present attacks that break FLoC's privacy properties. We first present a pre-image attack on SimHash (\cref{sec:ip-preimage-attack}) that breaks its pre-image resistance property. In our experiments, the pre-image attack can be used to mount a Sybil attack to break its $k$-anonymity property as well (\cref{sec:sybil-attack}). Using Generative Adversarial Networks (\emph{GANs})~\cite{goodfellow2020generative, guo2017long}, we propose the \emph{GAN-IP attack}, which recovers parts of the browsing history of real users, since GANs can be used to generate plausible browsing histories for users in a target cohort (\cref{sec:gan-ip}). The GAN-IP attack can reconstruct $10\%$ or more of the history in at least $30\%$ of the cases, breaking FLoC's guarantees of keeping browsing histories private. \cref{table:attacks-on-flocs} and \cref{fig:attacks-on-floc} summarize and illustrate the three attacks that we present.

\begin{figure*} 
    \centering
    \includegraphics[width=\linewidth]{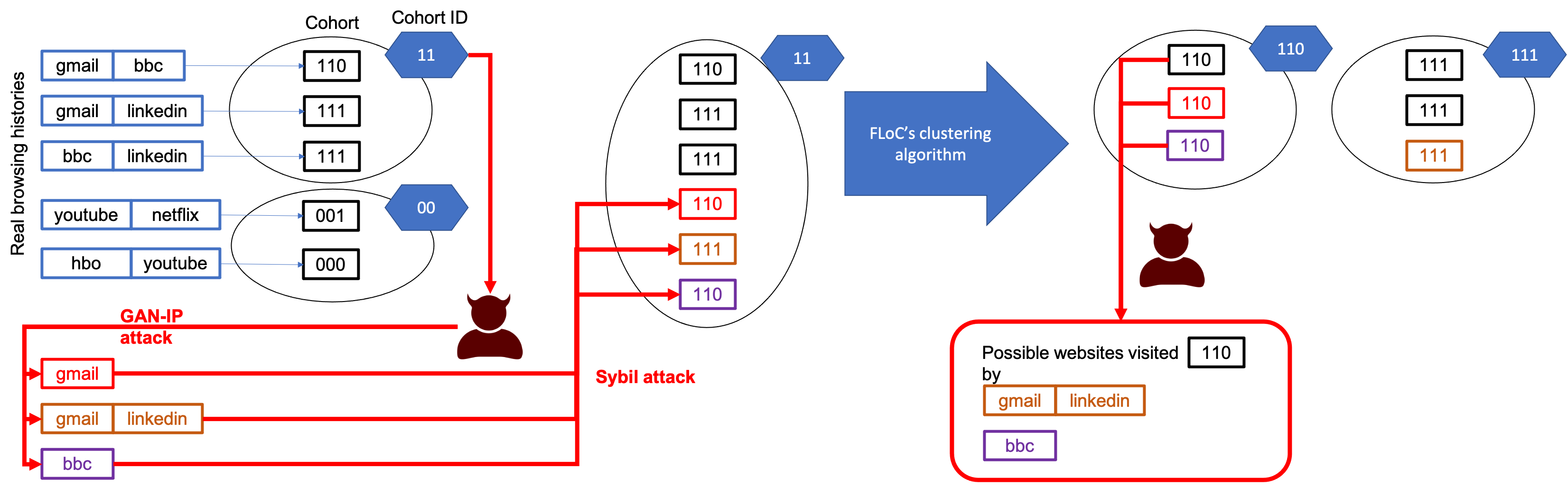} 
    \caption{How an attacker extracts private information from FLoC. First, the attacker takes a cohort ID $\gamma$ and then uses the GAN-IP attack to create fake browsing histories whose SimHash contains $\gamma$ as prefix. These SimHashes make the cohort decomposable, so  FLoC's clustering algorithm divides the cohort into smaller cohorts. The attacker exploits the fake histories to infer websites visited by real users.}
    \label{fig:attacks-on-floc}
\end{figure*}

\begin{table} 
    \centering
    \caption{Attack summary table}
    \label{table:attacks-on-flocs}
\resizebox{\columnwidth}{!}{%
    \begin{tabular}{c c c c} 
         \toprule
         Attack Name & Privacy Properties & Type of Attack \\ 
         \midrule
          Integer Programming & Pre-image Resistance & Pre-image Attack  \\ 
          Sybil  & $k$-anonymity & Forgery Attack\\ 
          GAN-IP & User Browsing Privacy  & Privacy Attack \\
         \bottomrule
    \end{tabular}
} 
\end{table}

We give an overview of the GAN-IP attack. Using a GAN we generate plausible user histories. Those histories are then given to an integer program. For each history, the integer program finds a non-empty subset of the history that matches the given target SimHash. \cref{fig:gen-ip-pipeline} illustrates the GAN-IP attack. We can optionally apply the GAN's discriminator on the integer program's output, as shown in the green frame in \cref{fig:gen-ip-pipeline}. In this way, the discriminator gives us a score on how realistic the produced browsing history is.

\begin{figure} 
    \centering
    \includegraphics[width=\linewidth]{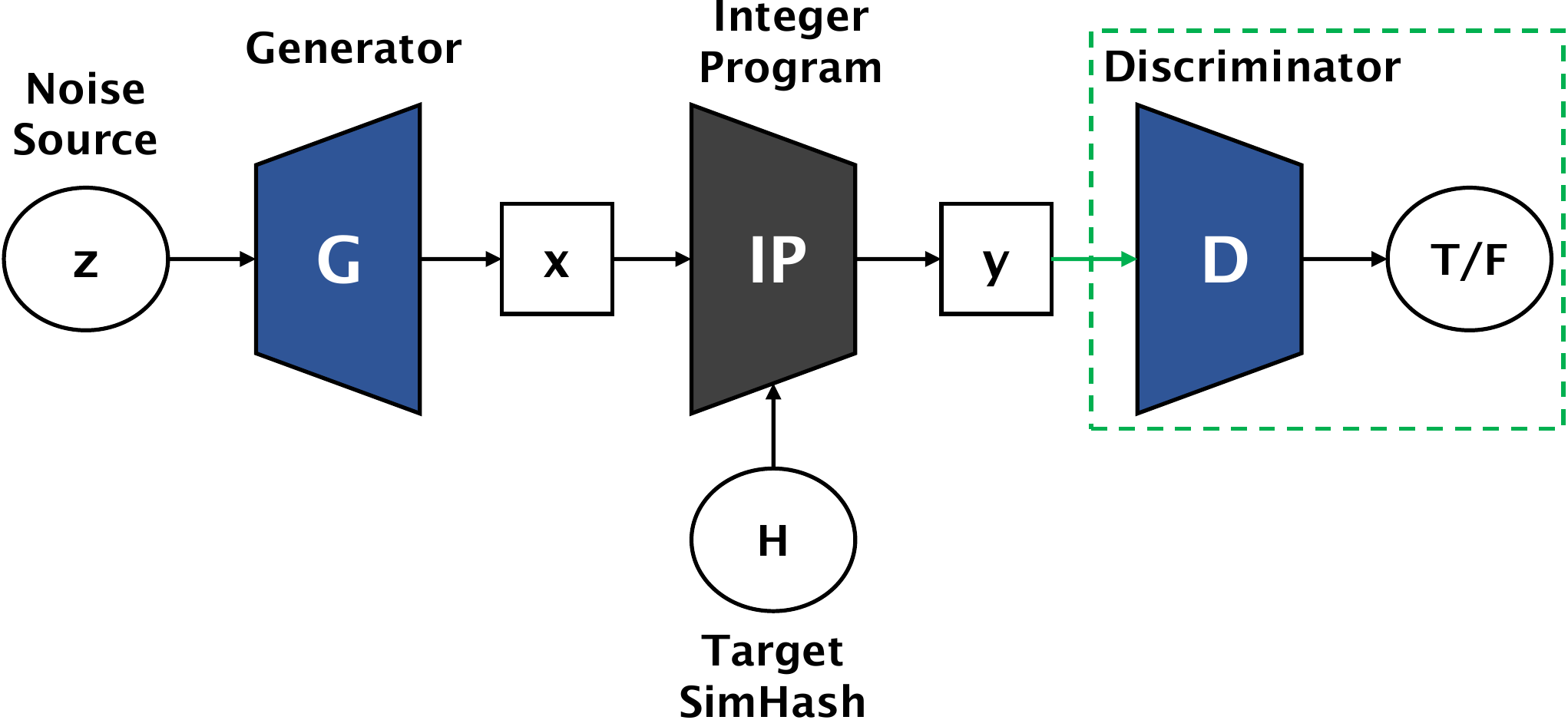} 
    \caption{Pipeline: integer programming on the generator outputs. The green boxed part is optional.}
    \label{fig:gen-ip-pipeline}
\end{figure}

\textbf{Attacker model.} We assume that the attacker's goal is to infer private information about the browsing history of a target user. For that, we assume the following capabilities. (1) The attacker has access to the FLoC implementation used by the users' devices. This is trivial since the code is embedded in the open-source Chrome browser. (2) The attacker can see the target user's FLoC ID, which the user sends to all websites embedding a FLoC request. (3) The attacker can actively create new users in the FLoC system. This is possible because the server that assigns cohort IDs takes as input only the SimHash and, therefore, it cannot distinguish genuine users from bots. (4) The attacker has access to the browsing histories of a sample of the user population, which can also be purchased from companies such as Comscore.\footnote{\url{https://www.comscore.com/}} Examples of such attackers, in an order of increasing capabilities, are operators of any website, tracking websites, and also Google itself.

We used the SimHash implementation from Chrome Version 91 according to capability 1. Our Sybil attack depends on knowledge of the target user's cohort ID (capability 2) and the ability to generate new users (capability 3). We train the model used for the GAN-IP attack on a publicly available dataset of movies, which the FLoC authors also used for evaluation. There exists also proprietary datasets of browsing histories that can be used in the real attack (capability~4).

\subsection{Integer Programming Pre-image Attack}\label{sec:ip-preimage-attack}

We now show how to compute pre-images of SimHashes using integer linear programming. Assume given a set $D = \left\{d_1, \ldots, d_n\right\}$ of domains (e.g., output by the GAN's generator) and a SimHash $z \in \mathbb{R}^\ell$ and we want to find a subset $D^* \subseteq D$ whose SimHash is $z$. We start by observing that $D^*$ must fulfill the following condition, by the definition of SimHash,
\begin{equation}
\text{sgn}\left(\sum_{d \in D^*}\eta_{d, j}\right) = z_j, \quad \text{for $j \leq \ell$},
\end{equation}
where $\eta_{d, j}$ is the $j$-th entry of $\eta_d$.
If we unfold the definition of $\text{sgn}$, this condition becomes the following:
for $j \leq \ell$, 
$\sum_{d \in D^*}\eta_{d, j} \geq 0$, if $z_j = 1$, and
$\sum_{d \in D^*}\eta_{d, j} < 0$, otherwise.
We can rewrite this condition as follows:
\begin{equation}
    (2z_j - 1) \sum_{d \in D^*}\eta_{d, j} \geq 0, \text{for $j \leq \ell$.}
\end{equation}

We see then that finding a pre-image of the SimHash $z$ reduces to finding a subset $D^* \subseteq D$ that fulfills these $\ell$ inequalities. We now show how to do this using integer programming. We first represent subsets of $D$ as bitstrings in $\{0, 1\}^n$. A bitstring $x = (x_1, \ldots, x_n) \in \{0, 1\}^n$ denotes the subset $\left\{d_i : i \leq n, x_i = 1\right\}$. If $x^*$ is the bitstring representation of $D^*$, we can then rewrite the condition as:
\begin{equation}
    (2z_j - 1) \sum_{i \leq n}\eta_{d, j}x_i \geq 0, \text{for $j \leq \ell$.}
\end{equation}
This leads to the following linear integer program:
\begin{align}
    \max_x & \quad \sum_{i \leq n} x_i\\
    s.t. & \quad \left(2z_j - 1\right) \sum_{i \leq n} \eta_{d, j} x_i \geq 0, \text{ for $j \leq \ell$}\label{ip:const}\\
    & \quad x_i \in \{0, 1\}.
\end{align}
Note that by maximizing $\sum_{i \leq n} x_i$, we seek the largest subset $D^* \subseteq D$ that fulfills the conditions. Hence, this program searches for the largest subset of $D$ that yields the desired SimHash $z$. The maximization is also necessary to avoid outputting $x = 0^n$, which is a trivial solution. We summarize these insights with the following theorem.\looseness=-1

\begin{theorem}
Assume given a SimHash $z$ and the integer program $\mathit{IP}(D)$ above. If $x^* \in \{0, 1\}^n$ is an optimal solution to the program below and $x^* \neq 0$, then the SimHash of $x^*$ is $z$. 
\end{theorem}

As illustration we present the integer program with a history $D$ containing exactly \texttt{google.com}, \texttt{youtube.com}, and \texttt{facebook.com}. This is the history on the left of \cref{fig:simhash-comp-ex}
As a target SimHash, we choose the SimHash of the right history ($\underline{\mathbf{10111}}$). We get the following integer program. We maximize $\sum_{i \le 3 } x_i$ with the constraints
\begin{align*}
    \left(2\cdot\underline{\mathbf{1}} - 1\right) \cdot \left( -0.88\cdot x_1 + 1.11 \cdot x_2 + 1.61\cdot x_3\right)  &\geq 0 \\
    \left(2\cdot\underline{\mathbf{1}} - 1\right) \cdot \left( 0.62\cdot x_1 + 0.76 \cdot x_2 - 0.62\cdot x_3\right)  &\geq 0 \\
    \left(2\cdot\underline{\mathbf{1}} - 1\right) \cdot \left( 0.67\cdot x_1 - 0.26 \cdot x_2 - 1.55\cdot x_3\right)  &\geq 0 \\
    \left(2\cdot\underline{\mathbf{0}} - 1\right) \cdot \left( 0.18\cdot x_1 - 1.79 \cdot x_2 - 0.03\cdot x_3\right)  &\geq 0 \\
    \left(2\cdot\underline{\mathbf{1}} - 1\right) \cdot \left( 2.03\cdot x_1 - 1.51 \cdot x_2 - 0.07\cdot x_3\right)  &\geq 0.
\end{align*}
The optimal solution is $(x_1, x_2, x_3) = (1,1,0)$.
We conclude that the \texttt{facebook.com} domain of the history in the left-hand side of \cref{fig:simhash-comp-ex} must be removed to match the target SimHash for the history on the right-hand side. This means that the \texttt{netflix.com} domain of the right history is redundant, since it does not change the SimHash of the remaining domains.

\begin{table} 
    \centering
    \caption{Benchmark of GAN -- Integer Program}
    \label{table:attack-gan-ip}
    \begin{tabular}{c c c c} 
         \toprule
         SimHash Length & Success Rate & Int. Program Time \\ 
         \midrule
          5 & 100\% & $0.52$ s \\ 
          10 & 95\%  & $2.01$ s \\ 
          15 & 64\% & $5.03$ s \\
          20 & 34\% & $5.89$ s  \\
          25 & 11\% & $12.83$ s \\
         \bottomrule
    \end{tabular}
\end{table}

Although finding a pre-image of SimHash is NP hard, our integer programming attack is very efficient for the used bit lengths and history sizes, as illustrated in \cref{table:attack-gan-ip}. We vary the SimHash bit length from $5$ to $25$ in increments of $5$. Recall that in FLoC trials the SimHash length varied from $13$ to $20$ bits. For a given SimHash length, we sample a real history and compute its corresponding SimHash. The integer program then starts with a history $D$ of $32$ elements which can be either random or generated using a GAN introduced in \cref{sec:gan-ip}.

In \cref{table:attack-gan-ip}, the ``Success Rate'' column reports the percentage of histories generated by the GAN for which we could find a subset matching the target SimHash. We also report the average runtime in the ``Int. Program Time'' column. These results are based on executions on four different histories of real users generating at least 25 pre-image histories with the same SimHash.

This demonstrates that it is very efficient to find pre-images for a target SimHash. This facilitates the creation of fake users and the inference of private information the browsing history of real users.

\subsection{Sybil Attack}\label{sec:sybil-attack}

The privacy goals of FLoC is to achieve $k$-anonymity for the users~\cite{floc-whitepaper}. A Sybil attack floods a system with real users by generating fake (Sybil) entities. We show how using integer programming, we can mount a Sybil attack to deanonymize users hiding in clusters. The Sybil attack can isolate users in a cohort and identify them, breaking the $k$-anonymity property of FLoC.

Our Sybil attack works by observing a target user's cohort ID; that is, the substring $\sigma$ of the target user's cohort $C_\sigma$. Then we generate many users whose SimHash have $\sigma$ as prefix, called \emph{Sybil users}. They will all be assigned to the same cohort $C_\sigma$ by the clustering algorithm, described in \cref{sec:clustering-simhash}. In this way $C_\sigma$ eventually becomes decomposable and then it would be divided into $C_{\sigma0}$ and $C_{\sigma1}$. By repeating this, we can infer a sufficiently long prefix $\sigma'$ of the target user's SimHash and then observe the browsing histories of the Sybil users assigned to $C_{\sigma'}$ to obtain information about the target user's browsing history. Note that by creating sufficiently many Sybil users, we could ensure that $C_{\sigma'}$ consists only of a few real users and the rest only of Sybil users, making it easier to analyze their browsing histories. Also note that the cohort-assignment server has no mechanism protecting it from bots, since its only input is the SimHash.

In \cref{fig:sybil-illustration}, we demonstrate this attack on a toy example. At timestamp 1, we have the cohorts $C_0$ and $C_1$. The minimum size for a cohort is $k=2$. Mounting a Sybil attack to extend the prefix length of some cohorts, we generate two fake Sybil users, which are assigned to cohort $C_0$. The new Sybil users make $C_0$ $k$-decomposable, so the clustering procedure partitions $C_0$ into $C_{01}$ and $C_{01}$. Observe that $C_{01}$ consists now of two Sybil users and one real user. So the attacker can approximate the browsing history of that user using the generated browsing histories of the Sybil users.

\begin{figure} 
    \centering
    \includegraphics[width=\linewidth]{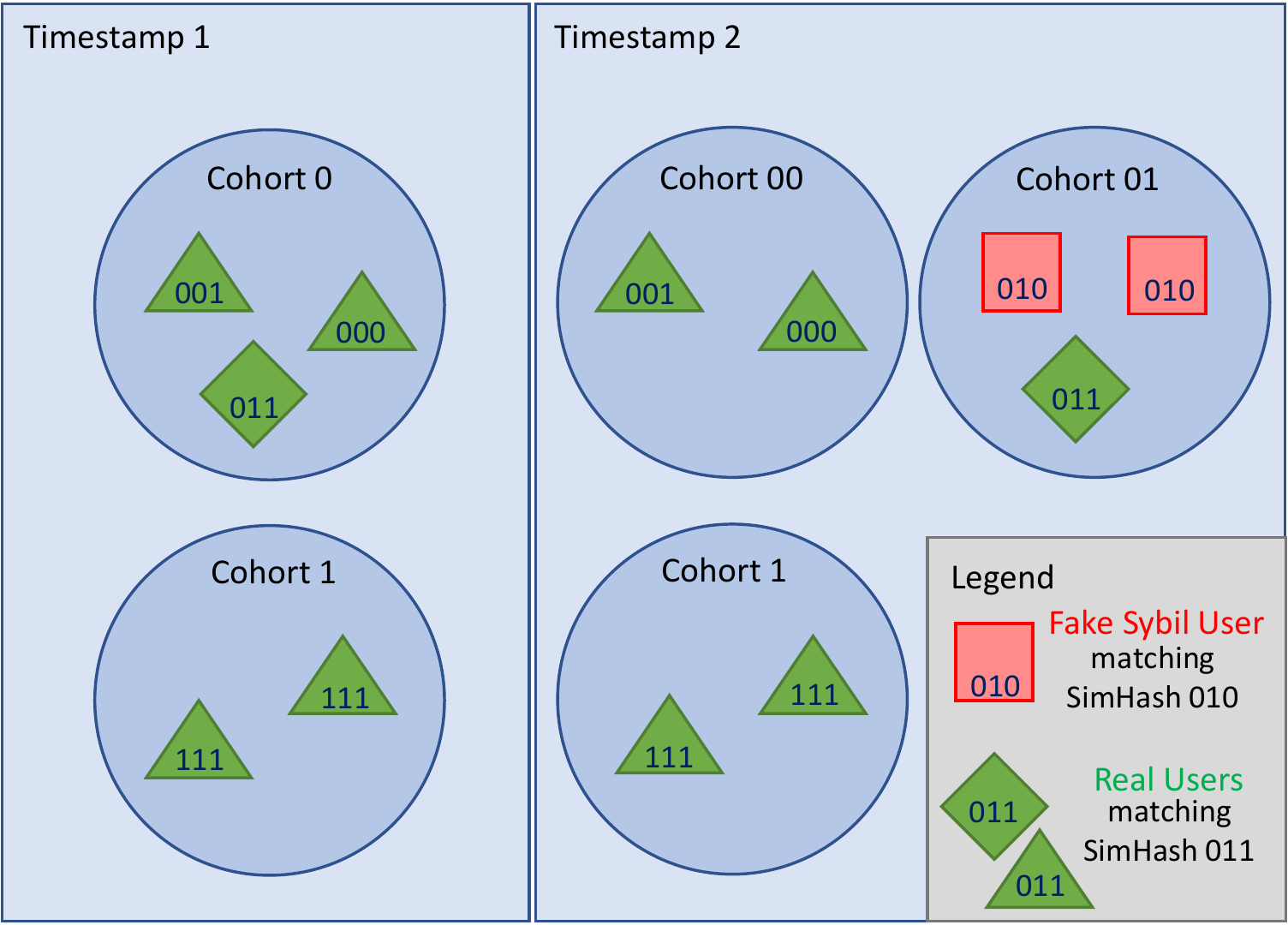} 
    \caption{Sybil attack example}
    \label{fig:sybil-illustration}
\end{figure}

\subsection{GAN-IP Privacy Attack}\label{sec:gan-ip}

The browsing histories generated by our integer program may not resemble a browsing history produced by a human. To produce a more realistic distribution of browsing histories and to gain more insights on the histories hidden in a cohort, we combine GANs with our integer programming attack to produce the GAN-IP attack.

Generative Adversarial Networks (GANs) can generate new samples from the same distribution as the training data.
A GAN consists of two neural networks, a generator \emph{G} and a discriminator \emph{D}. They compete against each other during training. The generator learns to produce realistic samples with the objective of deceiving the discriminator, while the discriminator learns to differentiate between the generated and real samples. From the several implementations available, we chose LeakGAN~\cite{guo2017long} because it is designed for text generation. However our attack works with essentially any GAN that can be adapted to produce users' histories.

We now present the GAN-IP attack. Suppose that we are given a SimHash $z$ of a given browsing history $h$ and that we want to produce a set $H$ of histories whose SimHashes are all equal to $z$. First, we use the LeakGAN to produce a set $H'$ of histories that resemble a sample from the distribution of browsing histories. Then for each $f \in H$, we attempt to compute a solution $x^*_f$ of $IP(f)$, the integer program induced by $f$. The desired set $H$ is $\left\{x^*_f : f \in H, IP(f) \text{ has a non-trivial solution}\right\}$. The attack is illustrated in \cref{fig:gen-ip-pipeline}.

To summarize, we can combine the three attacks above to extract private information as follows. First, we use a GAN to learn a distribution of users' histories such that, in approximately 30\% of the cases, the generated user will share 10\% or more of the history with the target user in the cohort. Using the GAN's generator, we can then produce fake browsing histories that look like histories from real users. Afterwards, we compute from this generated history a subset that matches a particular SimHash prefix of a target cohort using the IP-attack. These matching histories allow us to mount a Sybil attack, breaking not only k-anonymity of users, but extending the prefix length used to assign the cohort. This leaks more of users' SimHash, forming a self-reinforcing loop for the IP-attack, inferring parts of the users' browsing history.

\section{Attack implementation for FLoC}

For data protection reasons, we do not have access to a public browsing history dataset. To evaluate our attacks, we instead use the MovieLens dataset~\cite{Harper2015-cx}. An entry in this dataset contains movies watched by users over a period of time. Note that a movie history reflects a user's preferences and can be used to infer movie recommendations for that user. For these reasons, the MovieLens dataset acts as a good proxy to evaluate how our attacks would work in real browsing histories. We also remark that the FLoC's whitepaper also used this same dataset to evaluate FLoC~\cite{floc-whitepaper}.

We launch our GAN-IP attack on different movie histories from the MovieLens dataset. We demonstrate that the movie histories produced by our GAN-IP attack contain on average at least 10\% of the movie histories targeted by our attack. Furthermore, in about 50\% of our tests, the movie histories produced by the GAN component alone contain at least 10\% of the targeted histories. This demonstrates that the GAN-IP attack can extract information that was intended to remain private by the FLoC system.\looseness=-1

\subsection{Setup}

For demonstration purposes, the GAN was trained to produce movie histories with at most $32$ histories and using only the $5000$ most watched movies. However, our attack can be extended to larger movie histories and larger sets of movies.
We divided the MovieLens dataset into a training set and a test set. 
The training set contains $120\,000$ histories and the test set $5000$ histories.
The LeakGAN used by the GAN-IP attack was trained for $12$ hours using an Nvidia GeForce RTX $2070$ Super.

We evaluate our GAN-IP attack with 5 movie histories sampled from the test set. For each movie history $h_i$, with $i \leq 5$, we compute its SimHash $s_i$ and give it as input to the GAN-IP attack, which produces a set of movie histories $\hat{H}_i$ whose SimHash is also $s_i$. The set $\hat{H}_i$ contains at least 200 histories and $s_i$ is 15 bits long. We evaluate the quality of $\hat{H}_i$ with $I_i = \frac{1}{\left| \hat{H}_i\right|}\sum_{\hat{h} \in \hat{H}_i} \left|\hat{h} \cap h_i\right|$,  the average number of movies that the generated histories of the GAN-IP attack have in common with the target history $h_i$. The quality of our attack is then $q := \frac{1}{5}\sum_i I_i$. When reporting $q$, we also report the standard deviation of $I_1, \ldots, I_5$. We measure $q$ on various GAN models. As a baseline, we can use a random generator instead of the GAN’s generator. Note that $q$ indicates how much of the browsing history generated by our attack can be used to infer the movie history of a user with the same SimHash. Hence, our attacks shall maximize this value $q$.

\subsection{Results}

\textit{The GAN-IP attack can extract sensitive information from the SimHash.} \cref{table:eval-common-movies} reports $q$ for three versions of the GAN-IP attack: RAND, which uses only a random generator instead of a GAN to produce the set $H'$ of histories; GAN-41, which uses LeakGAN's weight from the saved training iteration 41; GAN-61, which is analogous to GAN-41 but for the later iteration 61. In parenthesis we give on average (in percentage) the part of the full generated history in common with the target history. For the GAN's generators, the average history length is approximately $27$ and $15$ after the integer program. For the random generator the numbers are $32$ and $17$. This sets the upper bound on the number of common movies, since the histories filtered by the integer program are only about half of the maximal length. Observe how GAN-41 produces higher values of $q$ than GAN-61 and RAND. Hence, stopping the training at iteration 41 yields histories with more movies in common with the target history.\looseness=-1

\begin{table} 
    \centering
    \caption{Distribution of Common Movie Counts}
    \label{table:eval-common-movies}
    \begin{tabular}{|c|c|c|c|c|} 
         \toprule
         Generator & \multicolumn{2}{|c|}{Common Movies $\pm$ stdev (\% of Gen. History Len.)} \\
         \hline
         & Generator & Int. Prog.\\
         \midrule
          RAND & $0.20 \pm 0.04$ ($<1\%$) & $0.17 \pm 0.03$ ($\approx1\%$)   \\
          GAN-41 & $\mathbf{2.39 \pm 0.96}$ ($\approx 9\%$) & $\mathbf{1.77 \pm 0.90}$ ($\approx 12\%$)   \\
          GAN-61 & $1.93 \pm 0.71$ ($\approx 7\%$) & $1.33 \pm 0.58$ ($\approx 9\%$)   \\
         \bottomrule
    \end{tabular}
\end{table}

\textit{The use of GANs significantly improves the attack's quality.} To demonstrate that the GAN-IP attack provides significant information, we compare the movies in the histories produced by RAND, GAN-41, and GAN-61. \cref{fig:cmovie-count-distrib-ip} is a histogram that shows, for $n \leq 11$ and each version of the GAN-IP attack, how many histories $h$ were produced such that $\left|h \cap h_i\right| = n$, for some $i \leq 5$. The number of movies that our generated histories have with the target histories is between 0 and 11. Observe how GAN-41 and GAN-61 in comparison with RAND have higher common movie counts with the target history. Therefore, histories generated with these GANs leak on average more information about the target history.

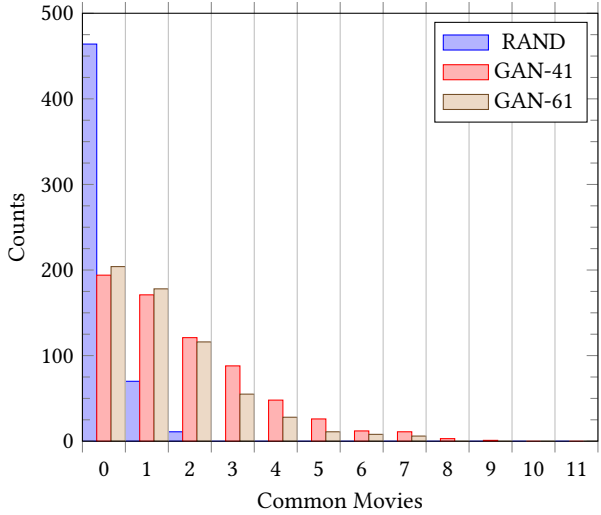
\begin{figure}
    \centering
    
    \begin{tikzpicture}
    \begin{axis}[
        ybar interval, 
        ymin=0, ymax=500, minor y tick num=3, 
        ylabel={Counts},
        bar shift=0pt, 
        xmin=0, xmax=12, 
        xtick={-1,0,...,13}, 
        xlabel={Common Movies},
        area style,
    ]
    \addplot coordinates { (0,464) (1,70) (2,11) (3,0) (4,0) (5,0) (6,0) (7,0) (8,0) (9,0) (10,0) (11,0) (12, 0) (13,0) } ;
    \addplot coordinates { (0,194) (1,171) (2,121) (3,88) (4,48) (5,26) (6,12) (7,11) (8,3) (9,1) (10,0) (11,0) (12, 0) (13,0)} ;
    \addplot coordinates { (0,204) (1,178) (2,116) (3,55) (4,28) (5,11) (6,8) (7,6) (8,0) (9,0) (10,0) (11,0) (12, 0) (13,0) } ;
    
    \legend{RAND, GAN-41, GAN-61}
    \end{axis}
    \end{tikzpicture}
    \caption{Histogram of common movie counts between $h$ and $\tilde{h}$ (with integer programming)}
    \label{fig:cmovie-count-distrib-ip}
\end{figure}

In \cref{fig:cmovie-count-distrib-gan}, we present an analogous histogram, but for the history produced only by the GAN. That is, we take the history $h'$ produced by the GAN before it was passed to the integer programming to produce the history $\tilde{h}$.\looseness=-1

\begin{figure}
    \centering
    
    \begin{tikzpicture}
    \begin{axis}[
        ybar interval, 
        ymin=0, ymax=500, minor y tick num=3, 
        ylabel={Counts},
        bar shift=0pt, 
        xmin=0, xmax=12, 
        xtick={-1,0,...,13}, 
        xlabel={Common Movies},
        area style,
    ]
    \addplot coordinates { (0,441) (1,91) (2,13) (3,0) (4,0) (5,0) (6,0) (7,0) (8,0) (9,0) (10,0) (11,0) (12, 0) (13,0) } ;
    \addplot coordinates { (0,104) (1,153) (2,117) (3,123) (4,67) (5,46) (6,30) (7,21) (8,8) (9,4) (10,1) (11,1) (12, 0) (13,0)} ;
    \addplot coordinates { (0,110) (1,154) (2,146) (3,80) (4,52) (5,37) (6,11) (7,11) (8,2) (9,1) (10,1) (11,1) (12, 0) (13,0) } ;
    
    \legend{RAND, GAN-41, GAN-61}
    \end{axis}
    \end{tikzpicture}
    \caption{Histogram of common movie counts between $h$ and $h'$ (without integer programming)}
    \label{fig:cmovie-count-distrib-gan}
\end{figure}
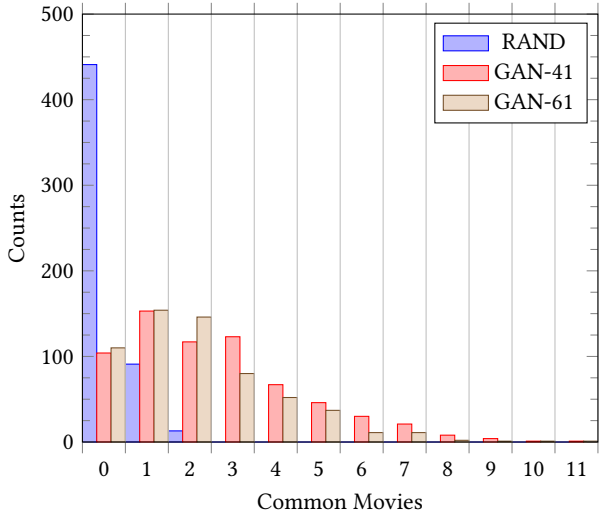

In \cref{fig:gan-output-example}, we show an example of a movie target history, a history generated by the GAN, and a history generated by our GAN-IP attack. The real history $h$ (on the left) is from the test set. We computed its SimHash and then generated a set of fake movie histories using the GAN. The history $h'$ in the middle is an example of such history. We gave this history as input to the IP-attack and then generated the history $\tilde{h}$ on the right. The SimHash of $\tilde{h}$ matches the SimHash of $h$ and 50\% of its movies are from $h$. In blue we show the movies that $h$ and $h'$ have in common.

\subsection{Discussion}

From the histograms we see that the random generation has very few movies in common with the target history. However, our GAN model evaluated at two different checkpoints has many more histories with a higher number of common movies. This is promising but only the tail of the distribution is on the higher counts, with a maximum of 9 common movies for one history of GAN-41 in \cref{fig:cmovie-count-distrib-ip}. In \cref{fig:cmovie-count-distrib-gan} both GAN-41 and GAN-61 have one history with 11 and 10 common movies with the target. Then GAN-41 has 4 histories with 9 common movies while GAN-61 still only has one. On average, the number of common movies with a target history is around 2 (see \cref{table:eval-common-movies}).

The histories generated by GAN-41 and filtered by integer programming have on average $12\%$ movies in common with the target histories. In around $28\%$ of the cases, the subset of movies selected by the IP attack matching the target SimHash reconstructs more than $10\%$ of the target history. Hence, the GAN-IP attack successfully breaks FLoC privacy claims and infers parts of the target user's history.\looseness=-1

Our attack can be amplified by performing it over a longer period of time. While the SimHash of users likely changes at every iteration (once per week), the majority of the browsing patterns remain. If our attack keeps generating the same movie over multiple runs, it increases the likelihood that the user watched the movie. Similarly, the union of the generated websites will more likely contain more movies that the user watched than just a single history.

Since Google runs the clustering algorithm, it is ideally suited to perform the GAN-IP attack. We therefore also note Google's capabilities that can make the attack more efficient. First, Google collects anonymized browsing histories of Chrome users that agreed with data collection in the Chrome User Experience Report. This gives them a significantly larger dataset of real browsing histories compared to the MovieLens dataset that we used. Second, Google has significantly more computational resources. Therefore, our results should be viewed only as a lower-bound of what a more powerful adversary can achieve.\looseness=-1

The generated history shares on average a non-negligible percentage of common movies with the unseen target history. The attack thus succeeds in revealing potentially sensitive information about the target user and, by extension, sensitive information about other users in the same cohort.

\section{MinHash Hierarchy}

In this section, we give some preliminaries on MinHash, a class of LSH (\cref{sec:background-minhash}). Afterwards, we present the MinHash Hierarchy system, a proposal for computing statistics on vehicles' trajectories (\cref{sec:application-minhash-hierarchy}). We then present our pre-image attack on the MinHash Hierarchy system (\cref{sec:attack-minhash-hierarchy}).

\subsection{MinHash} \label{sec:background-minhash}

MinHash is a type of LSH proposed by Broder~\cite{broder1997resemblance}. For a set of objects $\mathcal{X}$, MinHash estimates the similarity of subsets of $\mathcal{X}$. A MinHash is a function $h$ that maps each subset $X$ of $\mathcal{X}$ to a pseudo-random sequence $h(X) = (s_1, \ldots, s_n)$ of $n$ bitstrings. Usually, these bitstrings have 32 bits of length.
The function $h$ has the following property: for any $Y \subseteq \mathcal{X}$, the probability of $h(X) = h(Y)$ is the Jaccard similarity between $X$ and $Y$, i.e., $\frac{\left|X \cap Y\right|}{\left|X \cup Y\right|}$.

A MinHash function $h$ is composed of $n$ hash functions $h_i: \mathcal{X} \to \mathbb{N}$ and the MinHash of $X \subseteq \mathcal{X}$ is $h(X) = (s_1(X), \ldots, s_n(X))$, where $s_i(X) = \min_{x \in X} h_i(x)$.

A common choice for each $h_1, \ldots, h_n$ builds upon a hash function $\pi: \mathcal{X} \to \{0, \ldots, 2^{32} - 1\}$ that maps $\mathcal{X}$ to the set of 32-bitstrings. Then, for $i \leq n$, $h_i(x) = r \cdot \pi(x) + c \mod p$, where $r$, $c$, and $p$ are chosen uniformly at random from a sufficiently large interval of natural numbers and $p$ is a prime number larger than $\max \{\pi(x) : x \in \mathcal{X}\}$~\cite{broder1997resemblance}.

We illustrate the computation of a MinHash signature on a simple example. We define three hash functions $h_1(x) = x+3 \mod 5$, $h_2(x) = 2x+1 \mod 5$, and $h_3(x) = 3x+4 \mod 5$. Let $\mathcal{X} = \{0, 1, 2, 3, 4\}$. We now compute the MinHash signature $h(X)$ for the set $X=\{1, 4\}$. Note that $s_1(X) = \min\{h_1(1), h_1(4)\} = 2$, $s_2(X) = \min\{h_2(1), h_2(4)\} = 3$, and $s_3(X) = \min\{h_3(1), h_3(4)\} = 1$. Hence, the MinHash signature for the set $X$ is then $(2,3,1)$.

\subsection{Application of MinHash to Privacy}\label{sec:application-minhash-hierarchy}

In this section, we present the MinHash Hierarchy~\cite{minhashprivacy}, which is a proposal for computing statistics on mobile entities' trajectories. One example of such a statistic is the most popular route in the city. The MinHash Hierarchy can compute such statistics by placing cellular base stations, called \emph{checkpoints}, in a city and assigning a bitstring to each vehicle. Each checkpoint collects the set $X$ of bitstrings of the vehicles that pass nearby, using mobile devices stored in the vehicle. Afterwards, each checkpoint stores \emph{a MinHash signature}, which is the MinHash of $X$.

We mainly focus on the MinHash aspect of the MinHash Hierarchy and we therefore simplify its explanation.

\subsubsection{MinHash Signatures}

We present here how the MinHash signatures are computed. Let $n$ be the number of vehicles circulating in a city. First, $m \ll n$ checkpoints are distributed over the city. Then $k \in \mathbb{N}$ hash functions $h_1, \ldots, h_k$ are fixed. The recommendation is to let $h_i(x) = ax + b \mod p$, with $i \leq k$, $a, b, p \in \mathbb{N}$, and $p > n$ prime, as shown before. However, if needed, cryptographic one-way functions can be used instead. 

Each checkpoint maintains a MinHash signature $s = (s_1, \ldots, s_k)$ so that, at any time, $s$ is the MinHash of the set of vehicles that passed by the checkpoint so far. To ensure this, $s_i$ is initially set to $\infty$, for $i \leq k$, as the MinHash of the empty set is $(\infty, \ldots, \infty)$. Next, whenever a vehicle whose assigned bitstring is $x$ passes by the checkpoint, $s_i$ is updated to $\min(s_i, h_i(x))$, for $i \leq n$.

With the checkpoints' MinHash signatures, Ding et al.~\cite{minhashprivacy} proposed the MinHash Hierarchy to efficiently perform popular path queries, such as finding the most frequented roads in a city during a given time interval. The process uses intersection and union operations defined for MinHash signatures of checkpoints to estimate the Jaccard similarities. Our attack focuses on the MinHash signatures and should work irrespective of the operations used to derive a given MinHash signature.

\subsubsection{Wrong Differential Privacy Claim}\label{sec:DP-counterexample} Ding et al.~\cite{minhashprivacy} claim in Theorem 5.1 that the MinHash Hierarchy provides differential privacy for the vehicles. We show that this claim is wrong. We start by recalling the definition of differential privacy.
An algorithm $A$ is $\epsilon$-differentially private if for any of $A$'s possible outputs $O$ and for all databases $D_1$ and $D_2$ that differ in only one individual~\cite{dwork2006calibrating}:

\begin{equation}
    P[A(D_1) = O] \le e^{\epsilon} \cdot P[A(D_2) = O].
    \label{eq:dp_minhash}
\end{equation}

In our context, a database $D$ is a set of \emph{trajectories}, each individual is a trajectory, and the Algorithm $A$ is the one used by a checkpoint to compute its MinHash signature. For simplicity and without loss of generality, we can assume a MinHash length of $k=1$; so there is only one single hash function $h$.

We refute Theorem 5.1 from~\cite{minhashprivacy} with the following counterexample. Let $D_1=\{t_1, \dots, t_n\}$ and let $D_2 = D_1 \setminus \{t_n\}$. Suppose that $h(t_1) > \dots > h(t_n)$. Therefore, $A(D_1) = h(t_n) < A(D_2)$, as $t_n \notin D_2$. Hence, $P[A(D_1) = h(t_n)] = 1$ whereas $P[A(D_2) = h(t_n)] = 0$. Since $e^\epsilon > 0$, for any $\epsilon \in \mathbb{R}$, Equation~\ref{eq:dp_minhash} cannot hold when $O = h(t_n)$.

\section{Attacks on MinHash Hierarchy}\label{sec:attack-minhash-hierarchy}

Our counterexample in \cref{sec:DP-counterexample} demonstrates that an attacker  with side knowledge can tell if a particular vehicle passed through a particular checkpoint. However, it does not tell us how much information it leaks in practice. Therefore, in this section, we present an attack breaking the privacy properties of the MinHash Hierarchy system that can be used directly to narrow down the area in which a vehicle traveled. In our experiments, we narrowed down the potential trajectory area to 10\% of the total area (in the number of checkpoints). 
The attack is illustrated in \cref{fig:minhash-attack}.

\begin{figure} 
    \centering
    \includegraphics[width=\columnwidth]{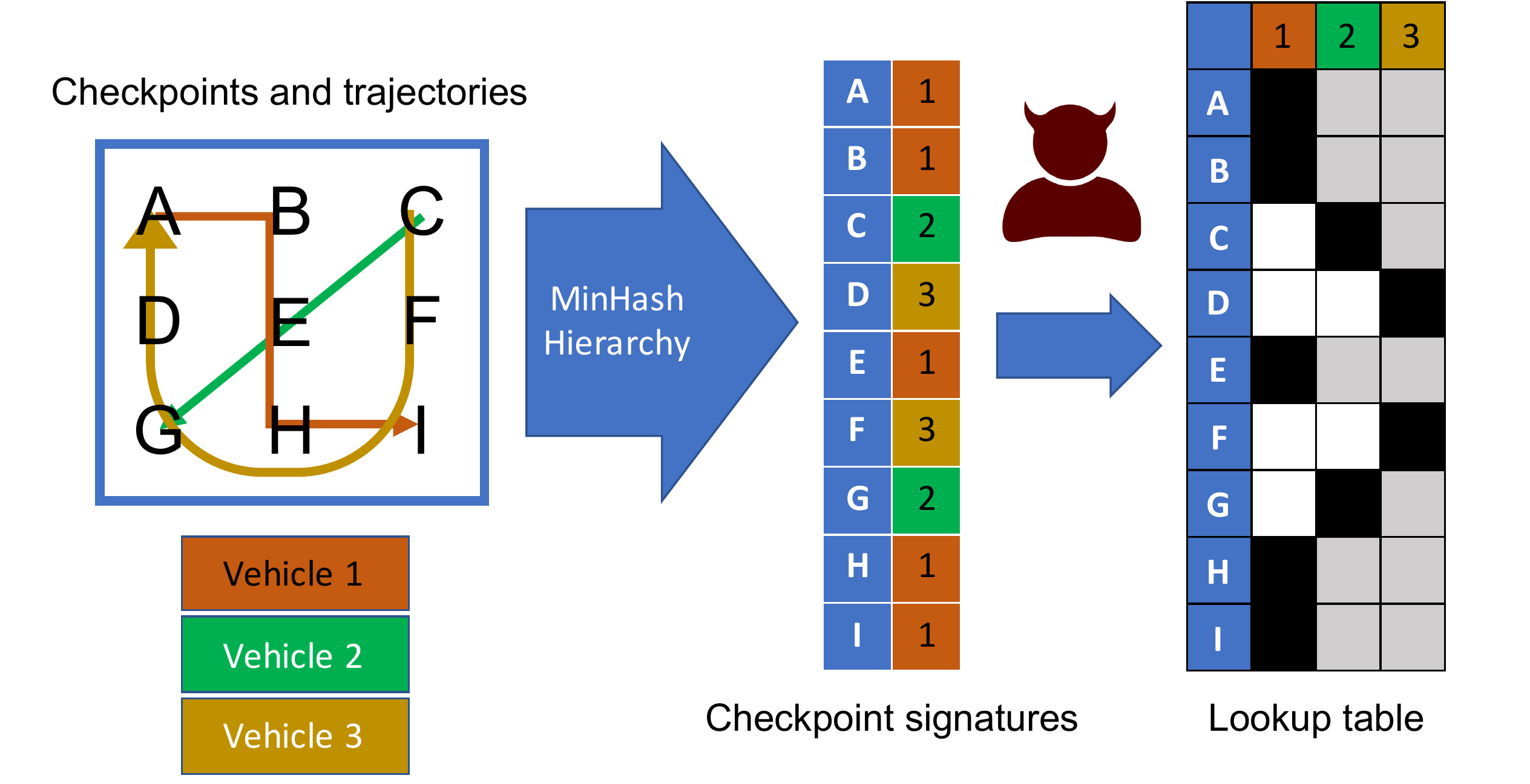} 
    \caption{How an attacker extracts private information from the MinHash Hierarchy. Left: Checkpoints A--I located in a grid and three vehicles' trajectories. Right: The checkpoints' signatures (we assume only one hash function: the identity function). The attacker computes a partial lookup table that says for each vehicle and each checkpoint, whether the vehicle passed or not by that trajectory (black: passed, white: not passed, gray: unknown). The lookup table reveals some checkpoints that were visited by some vehicles.}
    \label{fig:minhash-attack}
\end{figure}

\textbf{Attacker model} We assume that the attacker wants to infer the trajectories of the vehicles whose data is collected by the MinHash Hierarchy. We also assume that the attacker can access each checkpoint's signature, knows the hash functions used to compute the signatures, and can efficiently compute collisions for them. 

Note that the MinHash Hierarchy fulfills the requirements above. Furthermore, the hash functions used by the MinHash Hierarchy are just permutations. Due to this implementation choice, we can not only compute signature collisions, but we can also invert the permutation, extracting the user identifier. Should it instead use cryptographic hash functions with a large hash length like 256 bits, it would still be possible to precompute a look-up table with the hashes of all users. This is because the input space is the set of all mobile users and the cardinality of this user set is small.

We now present our attacks. Suppose that we are given a vehicle identified with bitstring $v$ and let $z = (z_1, \ldots, z_n)$ with $z_i = h_i(v)$, for $i \leq n$, be \emph{the signature of $v$}. Let $C$ be the set of checkpoints in a geographical area. For a checkpoint $c \in C$, we denote its signature with $s(c) = (s_1(c), \ldots, s_n(c))$. We use Algorithm~\ref{algo:divide_checkpoints} to partition $C$ into three subsets $B_z$ (black), $G_z$ (gray), and $W_z$ (white). $W_z$ denotes all checkpoints $c$ such that $z_i < s_i(c)$, for some $i \leq n$. Note that this condition means that the vehicle is not in the set of vehicles that passed through $c$. Otherwise, $s_i(c) \leq z_i$. $B_z$ contains all the points not in $W_z$ such that $z_i = s_i(c)$, for some $i \leq n$. Note that if $c \in B_z$, then it is very likely, except for a rare hash collision, that the vehicle passed through $c$. Finally, $G_z$ contains all other checkpoints in $C$: checkpoints not in $W_z$ for which $z_i < s_i(c)$, for all $i \leq n$. Note that if $c \in G_z$, then it is still likely that the vehicle passed through $c$, but not as likely as if $c \in B_z$.

\begin{algorithm}
\caption{Attack on MinHash Hierarchy}
\label{algo:divide_checkpoints}
\begin{algorithmic}[1] 
\State{$W_z \gets \emptyset$}
\For{$i = 1, \ldots, n$}
    \For{$c \in C$}
        \If{$z_i = s_i(c)$}
            \State{$B_{z} \gets B_{z} \cup \{c\}$}
        \EndIf
        \If{$z_i < s_i(c)$}
            \State{$W_z \gets W_z \cup \{c\}$}
        \EndIf
    \EndFor
\EndFor
\State{$B_z \gets B_z \setminus W_z$}
\State{$G_z \gets C \setminus (B_z \cup W_z)$}
\State{\textbf{return} $W_z, G_z, B_z$}
\end{algorithmic}
\end{algorithm}

\begin{theorem}\label{thm:ckpt-traj}
Let $v$ be a vehicle with signature $z$ and let $c$ be a checkpoint. 

\begin{itemize}
    \item If $c \in W_z$ then $v$ cannot have passed through $c$.
    \item If $c \in B_z$ then $h_i(v) = h_i(v')$, where $i \leq n$ and $v'$ is some vehicle that passed through $c$.
\end{itemize}
\end{theorem}

\begin{proof}
If $c \in W_z$, then $z_i < s_i(c)$, for some $i \leq n$. Recall, by the definition of MinHash, $s_i(c) = \min\left\{h_i(v') : v' \in V_c\right\}$, where $V_c$ denotes all vehicles that passed through $c$. Hence, $v \notin V_c$; otherwise, $s_i(c) \leq h_i(v) = z_i$, which is a contradiction.

For the second claim, note that if $c \in B_z$, then $z_i = s_i(c)$, for some $i \leq n$. By the definition of MinHash, we have $h_i(v) = z_i = s_i(c) = h_i(v')$, for some vehicle $v' \in V_c$.
\end{proof}

We emphasize that the hash functions used by MinHash Hierarchy are not collision-resistant. Even if they use cryptographic hash functions with a large length, note that the set of mobile users is relatively small enough so that one can precompute a look-up table with the hashes of all users. Therefore, if $c \in B_z$, then $v$ is likely to have passed through $c$.

Observe that $G_z \cup B_z$ describes all possible checkpoints the vehicle could have visited. In our experiments, we found that in average $G_z \cup B_z$ contains only around 10\% of all checkpoints in $C$.

To illustrate this attack, consider a scenario with 20 vehicles' trajectories, 5 checkpoints and 2 hash functions. The vehicle MinHash is $z=(z_1, z_2)=(9, 11)$. We compare each checkpoint's signature entries to the corresponding vehicle hash. 
\begin{table}
\centering 
\caption{Example signatures for several vehicles and checkpoints \label{tab:ckpt-sig}}
\begin{tabular}{ c | c | c | c }
\toprule
     $c$ & $\left(s_1(c), s_2(c)\right)$ & $v$ & $\left(z_1(v), z_2(v)\right)$  \\ %
     \midrule
     $c_1$ & $(8, 12)$ & $v_1$ & $(9, 11)$ \\  
     $c_2$ & $(6, 3)$ & $v_2$ & $(2, 8)$ \\  
     $c_3$ & $(2, 7)$ & $v_3$ & $(12, 13)$ \\  
     $c_4$ & $(4, 11)$ & $v_4$ & $(7, 10)$ \\  
     $c_5$ & $(11, 5)$ & $v_5$ & $(5, 18)$ \\
\bottomrule
\end{tabular}
\end{table}
Considering vehicle $v_1$ and the 5 checkpoints in \cref{tab:ckpt-sig}, our attack returns $W_z = \{c_1, c_5\}$ because at least one of the hashes is greater in the checkpoint signature. $B_z = \{c_4\}$ because the checkpoint is not in $W_z$ and at least one hash is equal. $G_z = \{c_2, c_3\}$ contains the remaining checkpoints. 

\subsection*{Identifying Vehicles}

Algorithm~\ref{algo:divide_checkpoints} takes as input a trajectory and identifies the checkpoints that could have been visited in that trajectory. It is also possible to modify this algorithm so that the input is a checkpoint and the output is the subset of vehicles from a set $V$ that potentially visited that checkpoint. The result is Algorithm~\ref{algo:checkpoint_calc}. This algorithm produces, from a given checkpoint $c$, three sets of vehicles: $W_c$, containing the vehicles that certainly did not pass through $c$; $B_c$, the vehicles that most likely passed through $c$ (except in the rare case of a hash collision); and $G_c$, containing the remaining vehicles. For example, suppose that we run this algorithm with checkpoint $c_4$ as input and with $V$ as the 5 vehicles listed in \cref{tab:ckpt-sig}. Then $W_c=\{v_2, v_4\}$ since $z_1(v_2) < s_1(c_4)$ and $z_2(v_4) < s_2(c_4)$. $B_c=\{v_1\}$ since $z_2(v_1) = s_2(c_4)$. Finally, $G_c=\{v_3, v_5\}$ contains the remaining checkpoints.\looseness=-1

\begin{algorithm}
\caption{Estimating vehicles passing through $c$}
\label{algo:checkpoint_calc}
\begin{algorithmic}[1] 
\State{$W_c \gets \emptyset$}
\For{$i = 1, \ldots, n$}
    \For{$v \in V$}
        \State{Compute $z = (z_1, \ldots, z_n)$ with $z_i = h_i(v)$}
        \If{$z_i = s_i(c)$}
            \State{$B_{c} \gets B_{c} \cup \{v\}$}
        \EndIf
        \If{$z_i < s_i(c)$}
            \State{$W_c \gets W_c \cup \{v\}$}
        \EndIf
    \EndFor
\EndFor
\State{$B_c \gets B_c \setminus W_c$}
\State{$G_c \gets V \setminus (B_c \cup W_c)$}
\State{\textbf{return} $W_c, G_c, B_c$}
\end{algorithmic}
\end{algorithm}

\begin{theorem}\label{thm:vhcl-traj}
Let $v$ be a vehicle. 

\begin{itemize}
    \item If $v \in W_c$ then $v$ cannot have passed through $c$.
    \item If $v \in B_c$ then $h_i(v) = h_i(v')$, for some $i \leq n$ and $v'$ some vehicle that passed through $c$.
\end{itemize}
\end{theorem}

The proof is analogous to the previous one. Observe again, that for MinHash Hierarchy, if $v \in B_c$, then $v$ is likely to have passed through $c$ as one can easily precompute a look-up table with the hashes of all vehicles. This theorem shows that we can narrow the set of vehicles that passed through $c$ to the set $G_c \cup B_c$.

\section{Attack implementation for the MinHash Hierarchy}

In this section, we experimentally validate that our attack on MinHash Hierarchy substantially narrows down the set of possible checkpoints visited by a vehicle to approximately only 10\% of all checkpoints in the area.

\subsection{Dataset} 

The dataset used in the original paper~\cite{minhashprivacy} is not publicly available. We thus used another public dataset of vehicle trajectories~\cite{moreira2013predicting} for the city of Porto, Portugal. Each entry in the dataset defines a vehicle trajectory. The trajectory is described as a list of points, where each point is a pair containing the latitude and longitude of the taxi at a given time point.

\subsection{Methodology}

For our experiments, we take the first $n=30\,000$ trajectories in the Porto dataset. As some trajectories contain points that are far outside the city, we removed all points containing an extreme latitude or longitude. We defined a latitude as extreme if it was below 2\% or above 98\% of all latitudes in these trajectories. We defined a longitude as extreme analogously. We then created a set $C$ of $m=7744$ checkpoints by fitting an $88 \times 88$ square grid on all points in these trajectories.

To generate the MinHash signature for a vehicle, we compute the MinHash signature of a singleton set containing only the identifying number of the vehicle (taken in the $\{1,\ldots, 30\,000\}$ range) using $k=200$ hash functions. 
We then computed the checkpoints' signatures from the vehicles passing by, assuming every trajectory belongs to a different vehicle. Each vehicle's GPS coordinates in its trajectory generates one update for the closest checkpoint. Finally, we run our attack on MinHash Hierarchy and for the MinHash $z$ of each vehicle, we compute the sets $B_z$, $G_z$, and $W_z$. We then measure $A_z := \left|G_z \cup B_z\right|/\left|C\right|$, the ratio of checkpoints that our attack identifies as possibly visited by the vehicle to the total number of checkpoints. The quality of our attack is measured by how low $A_z$ is on average for all vehicles we tested. $A_z$ is around $10\% \pm 5\%$, showing that on average, we narrow down the set of checkpoints visited by the vehicle to only 10\% of all checkpoints in the map.

We execute this attack 5 times. Each time, we use a separate set of $30\,000$ different trajectories. 

\subsection{Results}

\cref{fig:all-ckpt-traj-pdf} shows a heatmap with $30\,000$ trajectories. Each pixel is a checkpoint and its brightness is proportional to the number of vehicles that visited that checkpoint.

\begin{figure} 
    \centering
    \includegraphics[width=\linewidth]{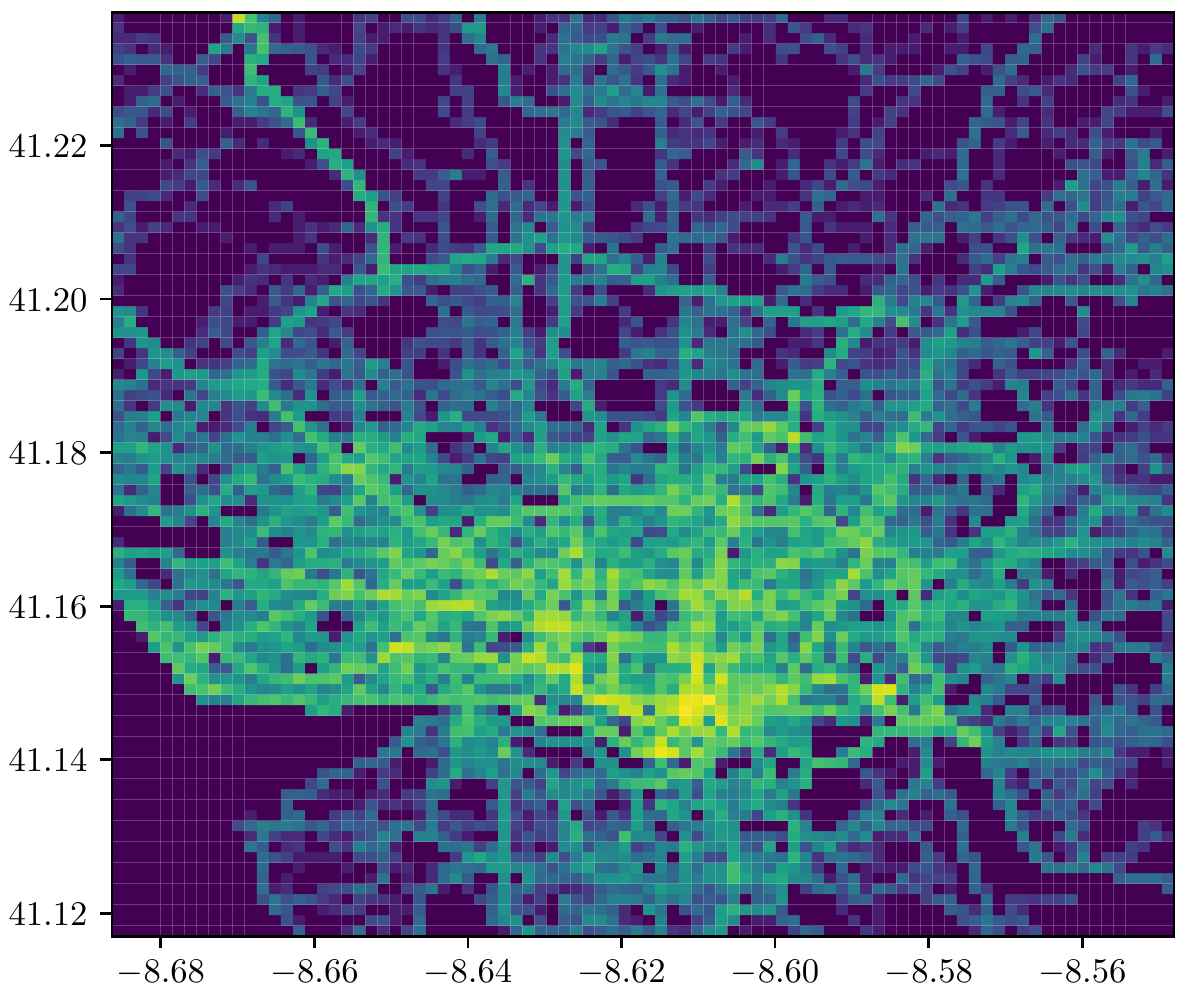} 
    \caption{Selected trajectories of checkpoint 2D histogram}
    \label{fig:all-ckpt-traj-pdf}
\end{figure}

\cref{subfig:good-traj-target} and \cref{subfig:bad-traj-target} show two example trajectories, using the square grid from \cref{fig:all-ckpt-traj-pdf}. The checkpoints visited by the vehicle are in black. \cref{subfig:good-traj-recov} and \cref{subfig:bad-traj-recov} show the outcome of our attack for these two trajectories, respectively. The checkpoints in $B_z$, $G_z$, and $W_z$ are marked black, gray, and white, respectively. 
For the recovery part, we note that trajectory B (28 checkpoints) is hidden within other trajectories (2018 checkpoints). However, some checkpoint signatures had values equal to the vehicle signature (16 checkpoints), and therefore those checkpoints are very likely to be part of the trajectory. For trajectory A, the attack can isolate the target (27 checkpoints) even more, narrowing it to only two possible trajectories. More trajectories are included in \cref{fig:example-traj} for illustration. These trajectories illustrate how, from only the checkpoints' MinHash signature, our attack can either accurately retrieve the target trajectories or restrict it to a much smaller area, thereby compromising the users' privacy.\looseness=-1

\begin{figure}
    \centering
        \subcaptionbox{Target A.\label{subfig:good-traj-target}}{\includegraphics[scale=0.32]{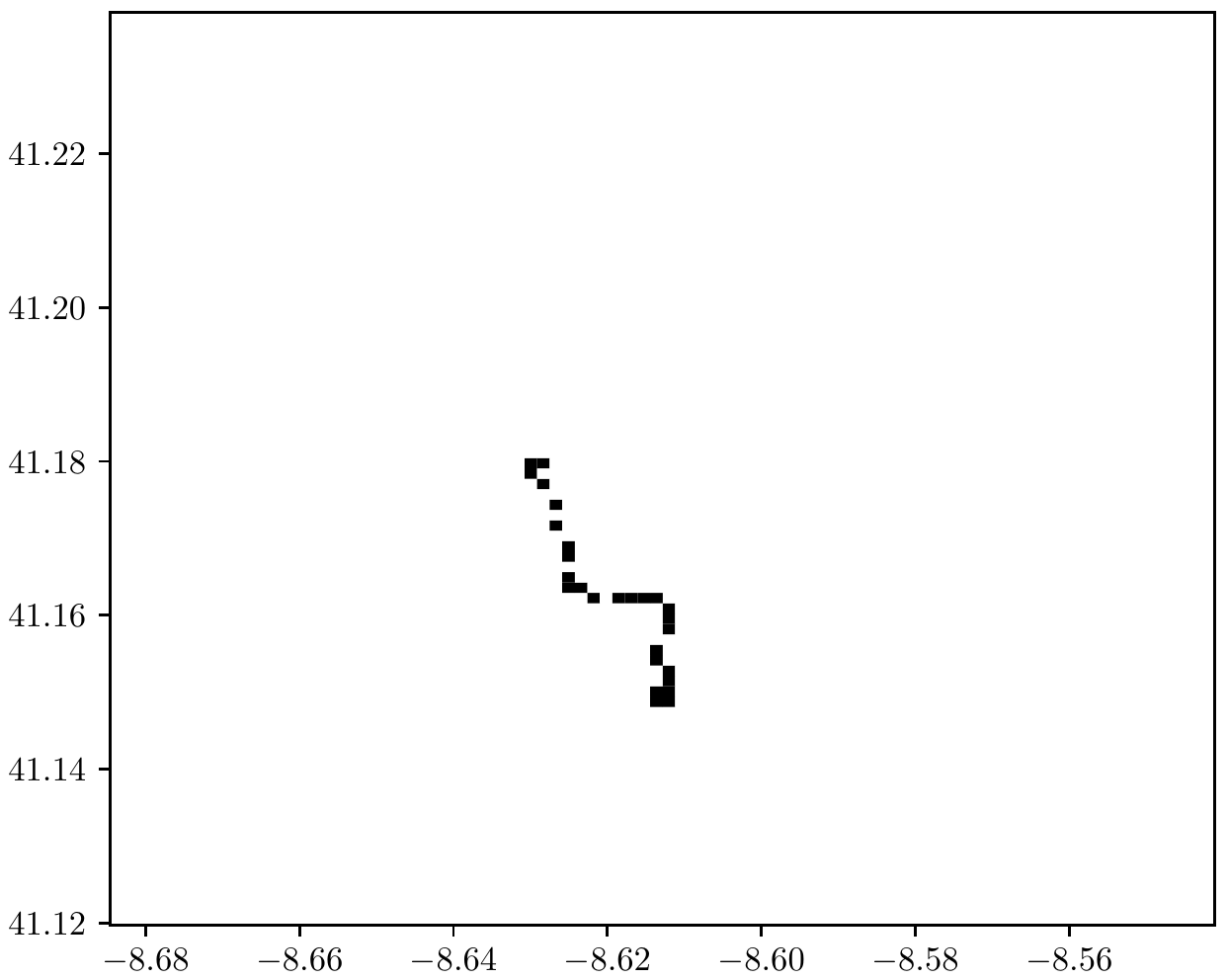}}
        \subcaptionbox{Recovered A.\label{subfig:good-traj-recov}}{\includegraphics[scale=0.32]{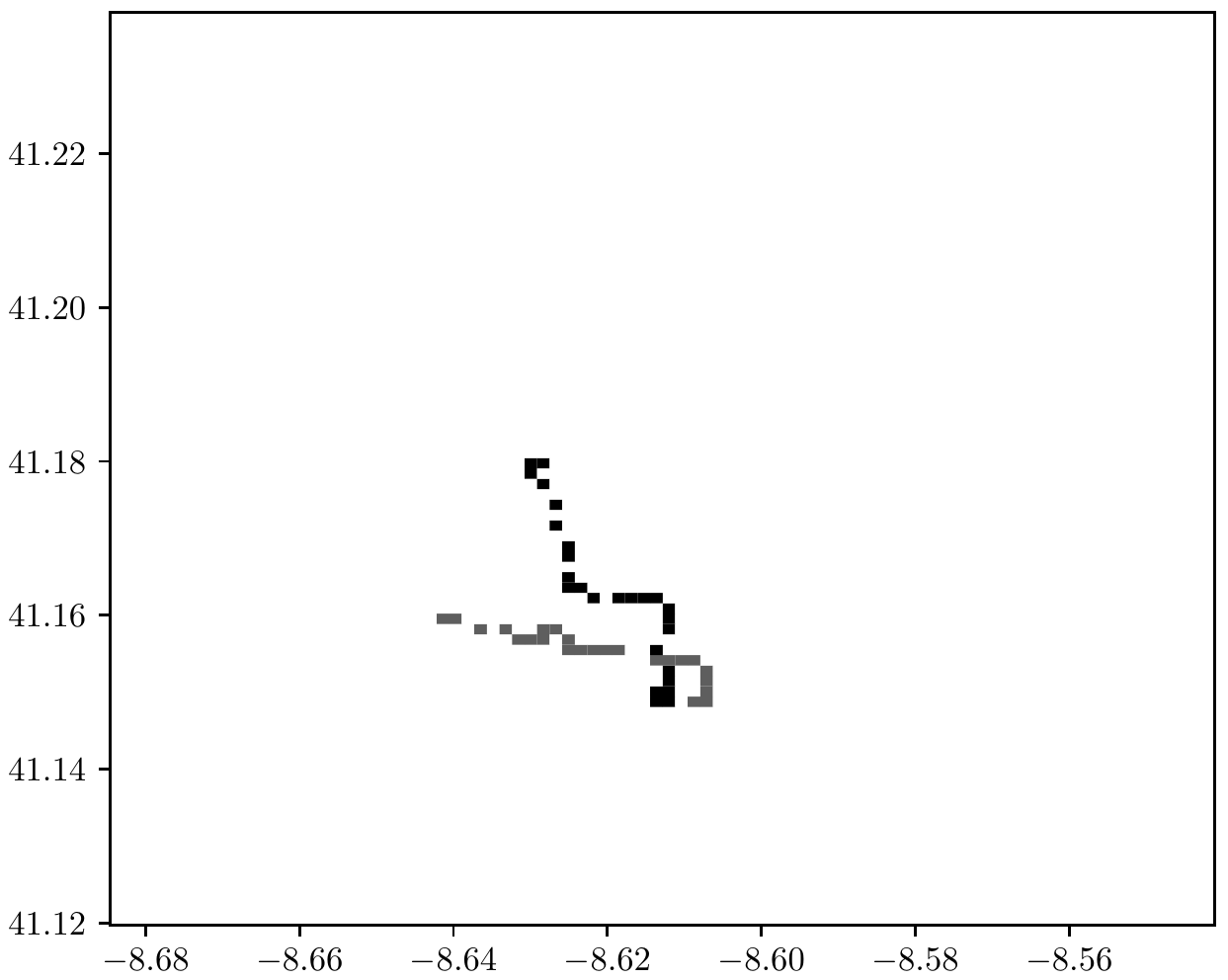}} \\
        \subcaptionbox{Target B.\label{subfig:bad-traj-target}}{\includegraphics[scale=0.32]{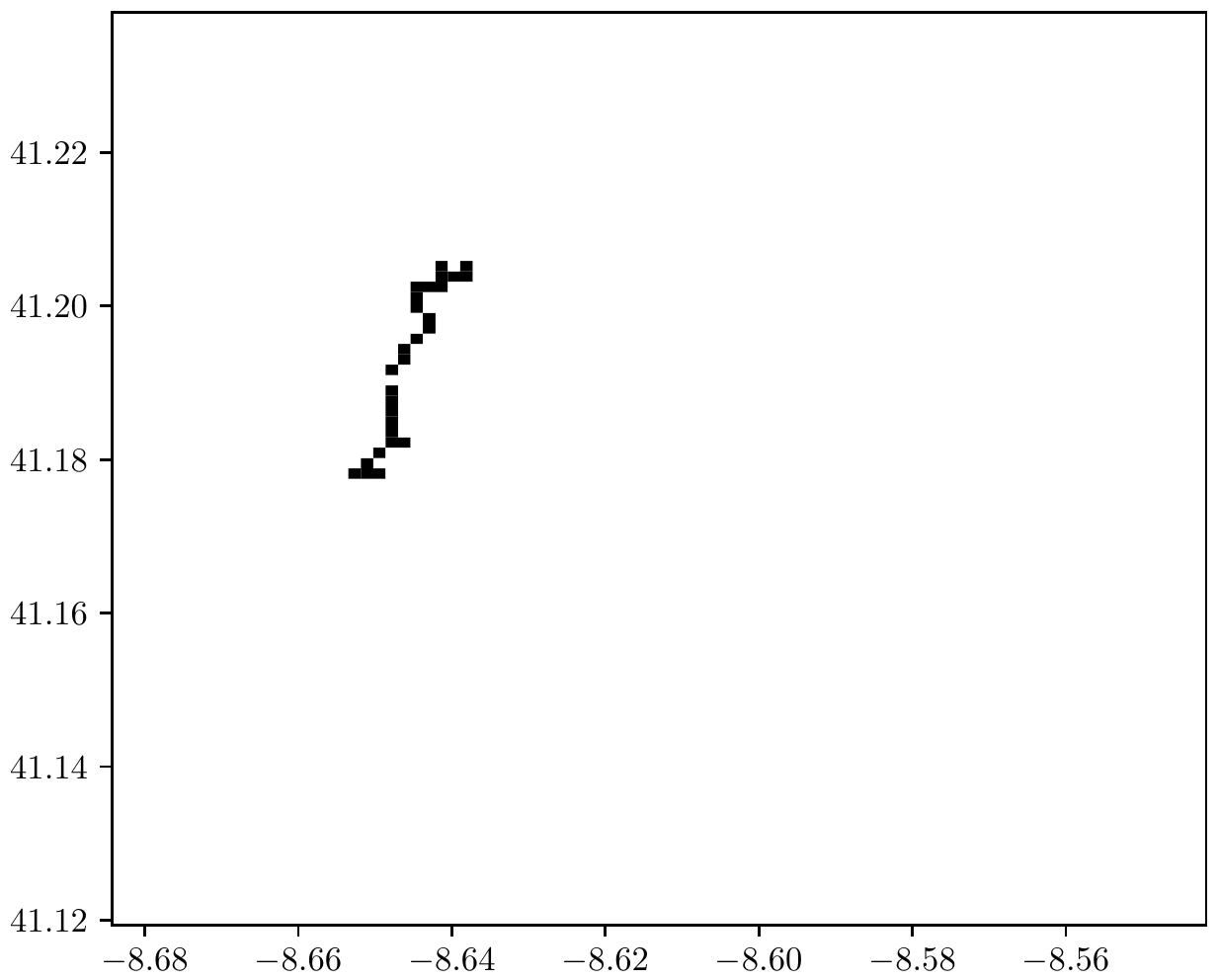}}
        \subcaptionbox{Recovered B.\label{subfig:bad-traj-recov}}{\includegraphics[scale=0.32]{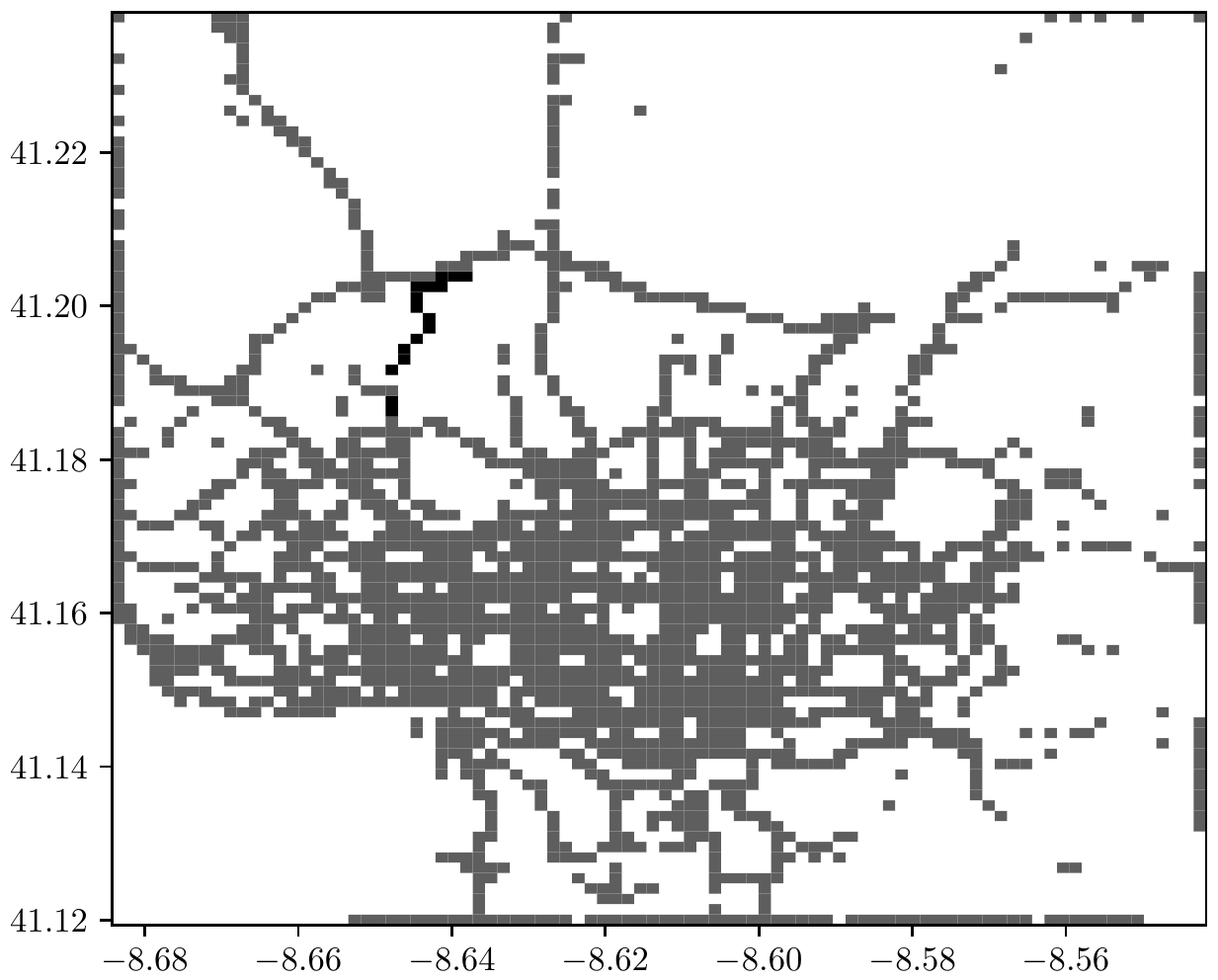}} \\
        \subcaptionbox{Target C.\label{subfig:traj-target3}}{\includegraphics[scale=0.32]{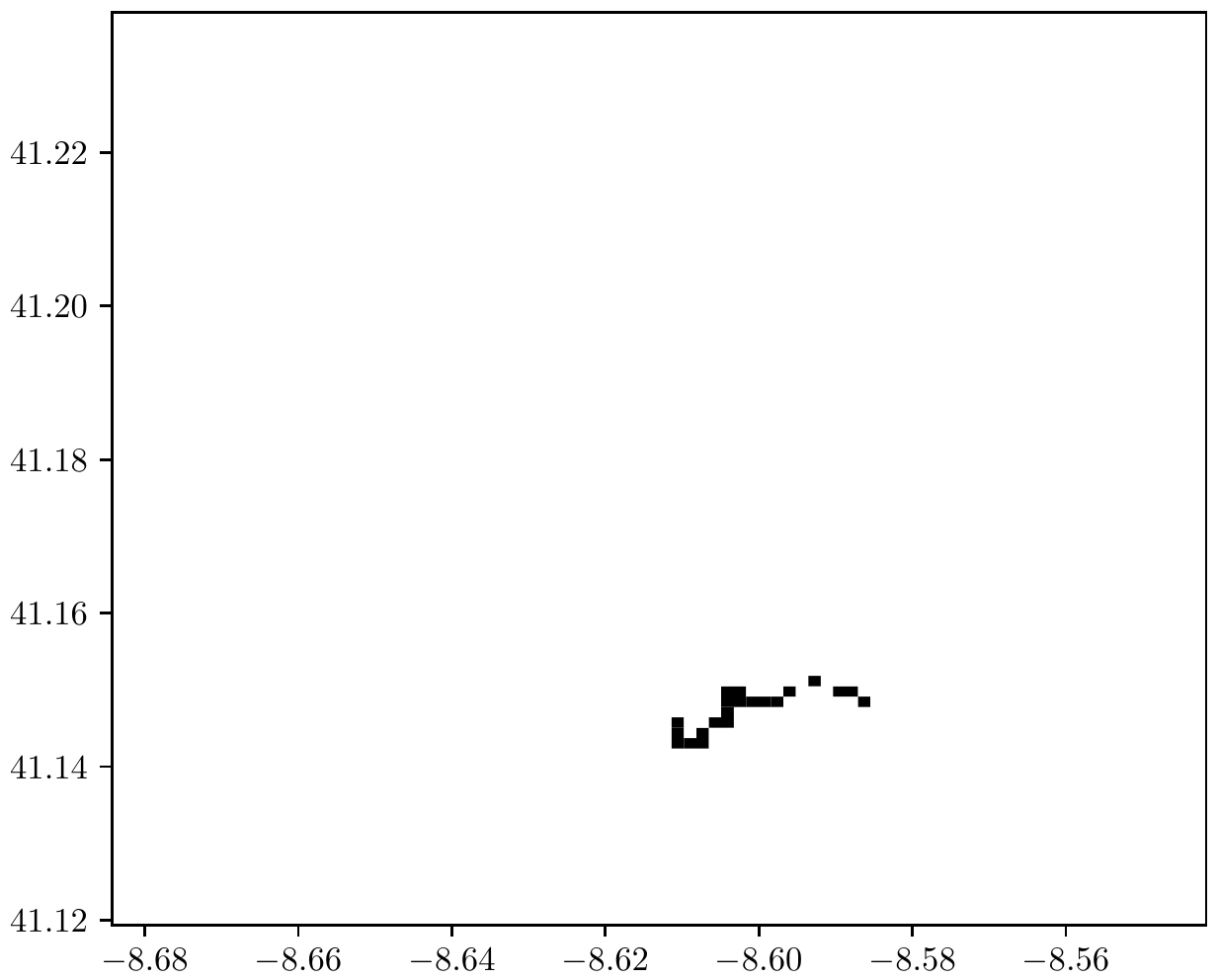}}
        \subcaptionbox{Recovered C.\label{subfig:traj-recov3}}{\includegraphics[scale=0.32]{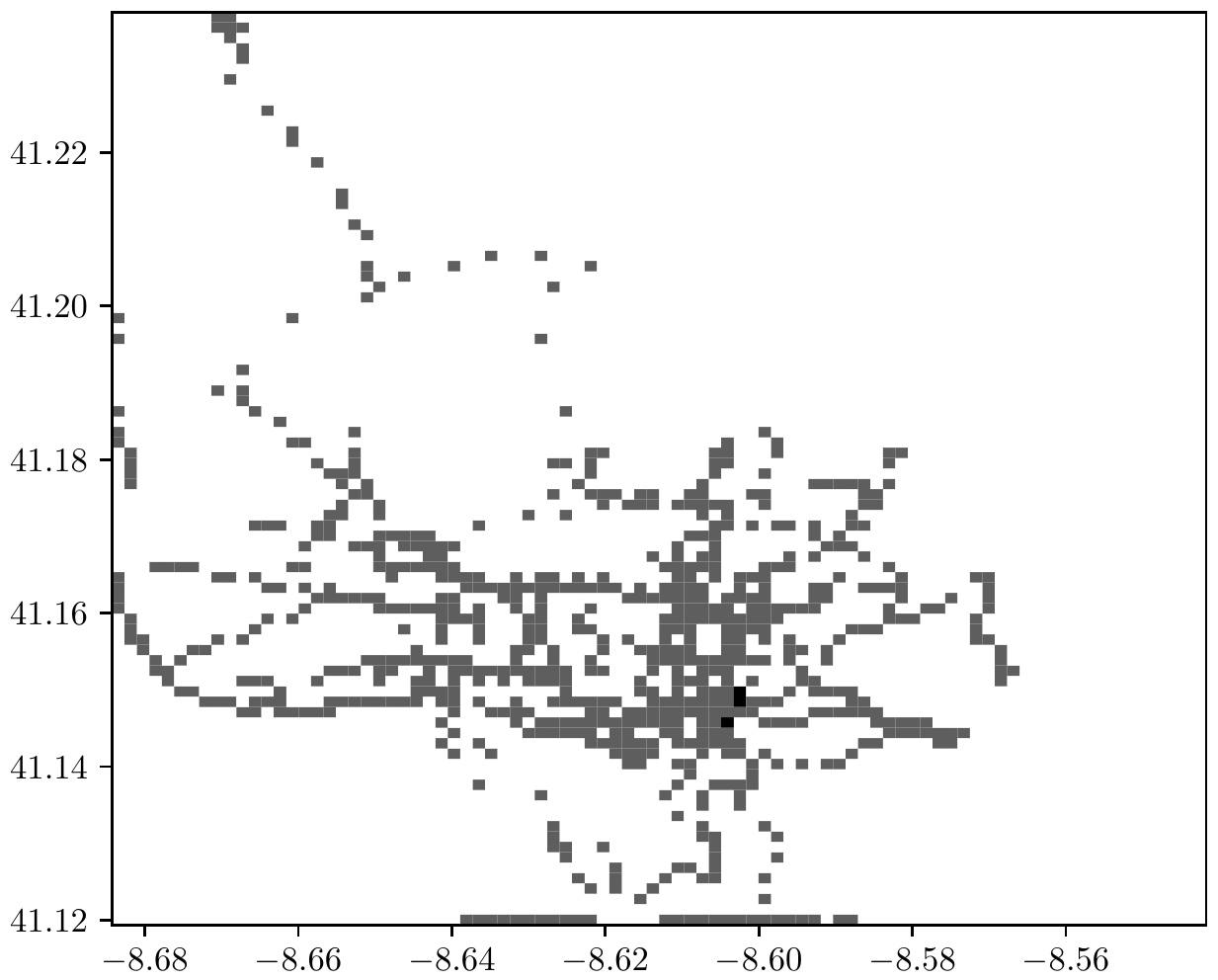}} \\
        \subcaptionbox{Target D.\label{subfig:traj-target4}}{\includegraphics[scale=0.32]{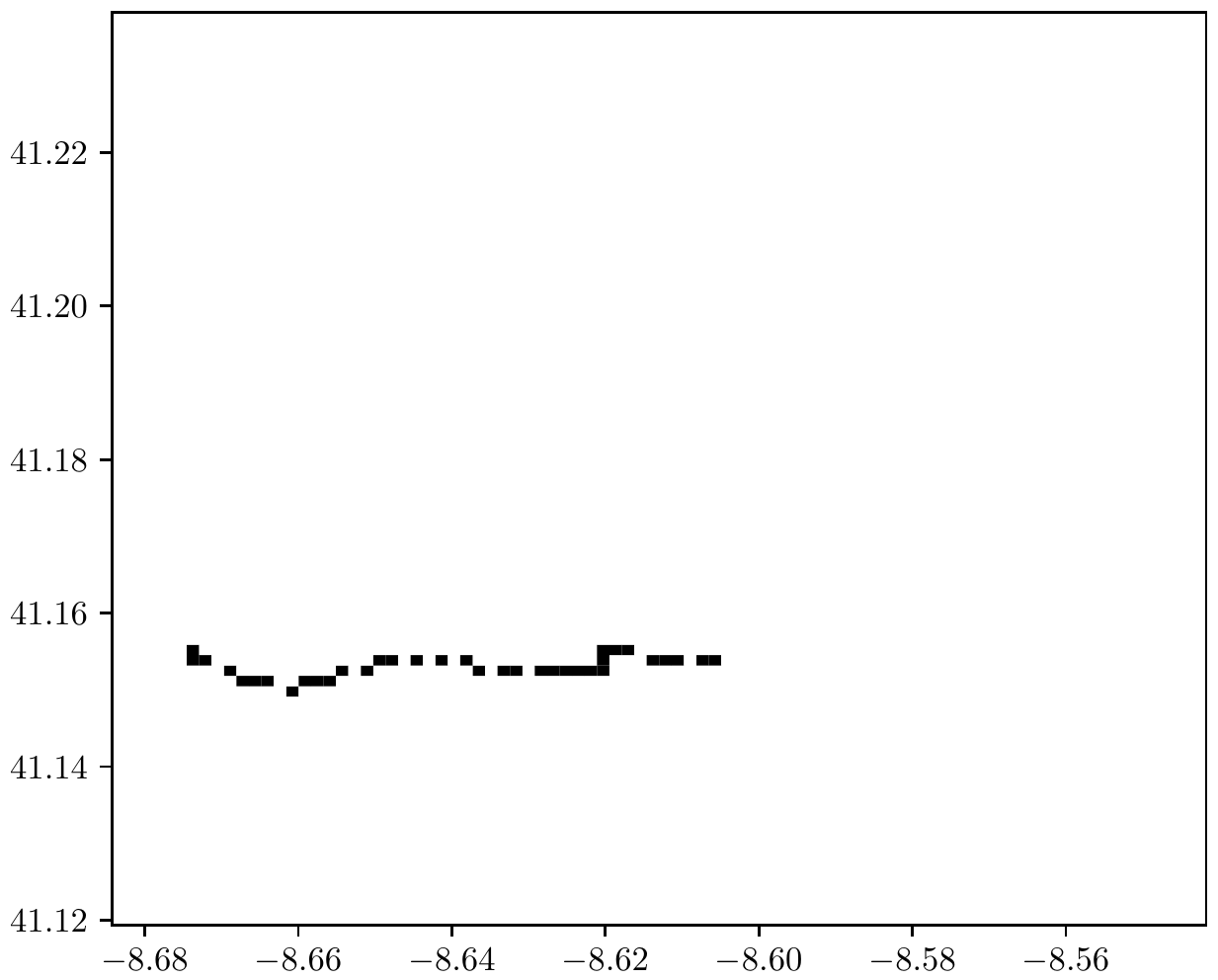}}
        \subcaptionbox{Recovered D.\label{subfig:traj-recov4}}{\includegraphics[scale=0.32]{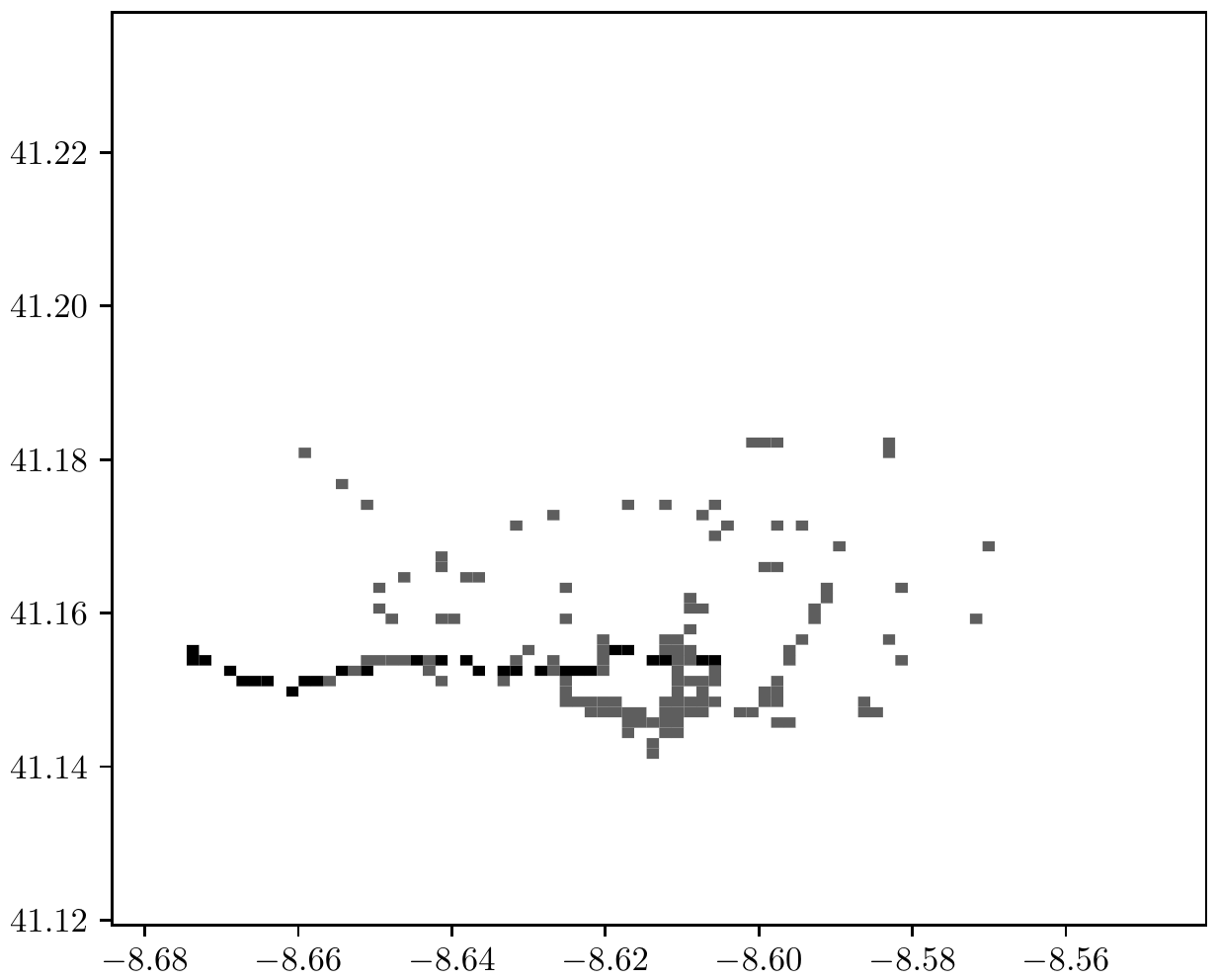}} \\
    \caption{Example of target and recovered trajectories}
    \label{fig:example-traj}
\end{figure}

In our dataset with $30\,000$ trajectories, a trajectory has in average $25.9 \pm 15.6$ checkpoints. A set of checkpoints found by our attack has on average $805.5 \pm 439.7$ checkpoints. Recall that the total number of checkpoint is 7744. Hence on average, we can reduce the set of possible checkpoints visited by a vehicle to around $10\%$ of the original set of checkpoints. This means that in a city like Porto we would restrict the trajectory to a neighborhood. Our attack breaks the MinHash Hierarchy's claimed privacy protection and shows how to confine the target trajectory to a small portion of the map.

\subsection{Discussion}

In contrast to the claims from Ding et al., we show that it is possible to isolate trajectories from checkpoint signatures with good accuracy, isolating a vehicle's potential trajectory to 10\% of the checkpoints on average.

If a checkpoint $c$'s MinHash signature contains a vehicle's hash, then we are certain (modulo the negligible probability of a hash collision) that the vehicle visited this checkpoint. Each checkpoint signature has 200 hash functions. We can therefore deanonymize up to 200 vehicles. In the city center, due to higher density of people, we can deanonymize a lower fraction of vehicles, while in a rural area we might be able to deanonymize all vehicles. Consequently, discarding trajectories with extreme latitudes or longitudes decreases the effectiveness of our attack, as our attack performs better for rural areas. This preprocessing was necessary to handle the given dataset.

As a post-processing step, it is possible to further filter the set $G_z$ based on time constraints, common commuting patterns, and other external information. We did not explore these techniques since we focus on showing the information leakage stemming from the application of LSH to sensitive data.

\section{Insufficiency of Countermeasures}\label{sec:mitigations} 

We now discuss some mitigations that have been proposed to address the privacy flaws in LSH systems. In particular, we argue that popular countermeasures like differential privacy are insufficient to prevent our attacks and that future work must avoid LSH systems as a way to offer privacy or develop stronger methods to provide privacy for these systems.

\textbf{FLoC.} For the FLoC proposal, one option is to make the Chrome operated server, which receives the SimHash and assigns the cohort IDs, trusted and unable to read sensitive information. This can be achieved, for example, by a trusted third party. However, this would prevent neither the Sybil nor the GAN-IP attack. This is because any party can observe a user's cohort ID, which is assigned based on a prefix of the user's SimHash. The Sybil attack can generate Sybil browsing histories that are mapped to that cohort. As the number of these Sybil histories grow, the users' cohort ID would become longer and reveal more of the user's SimHash, which enables the GAN-IP attack.

Recent work~\cite{google-practical-dp} proposes using differentially private clustering to compute the cohort IDs. It adds noise to each individual SimHash and then computes from the set $Z$ of all SimHashes a new smaller set $Z'$ of SimHashes, called a \emph{coreset}. Each SimHash $z$ in the coreset acts as a ``representative'' of a subset $S_{z}$ of SimHashes in $Z$ that are close to each other. The SimHash $z$ is computed using an additive-noise mechanism and comes with a positive number that approximately indicates the size of $S_{z}$. As a result, one cannot infer individual SimHashes in $Z$ from $z$.\looseness=-1

The use of coresets prevents us from conducting the pre-image attack on SimHashes from real individuals, as we cannot retrieve them. However, the SimHashes in the coreset are still vulnerable to the other attacks. In particular, we can still perform the Sybil attack on a SimHash in the coreset so that a SimHash there eventually represents mostly Sybil users. This would then isolate real users and we can then conduct the GAN-IP attack to infer parts of the browsing history of those users.

We argue that the best mitigation is a design of a new system that builds on differential privacy rather than $k$-anonymity to prevent Sybil attacks. This new system should also avoid leaking information in the LSH hashes. The new proposal of the Topics API~\cite{Topics-github} appears to satisfy both of these requirements but its implementation must still be formalized to allow a thorough evaluation of its privacy guarantees.

\textbf{MinHash Hierarchy.} For the MinHash Hierarchy, the use of differential privacy would provide guarantees about how much information any attack can extract. For example, with a low probability, if a checkpoint would compare its aggregated hash value with a random value instead of the current vehicle hash, it would lower the precision of our attack and give plausible deniability to vehicles.\looseness=-1

Recent works propose differentially-private versions for MinHash like PrivRec~\cite{zhang2020privrec} and PrivMin~\cite{PrivMin}. However, these proposals provide privacy only for the individual MinHashes and not for systems that process collections of MinHashes, like MinHash Hierarchy. Recall that each checkpoint computes the MinHash of the IDs of mobile devices that pass near the checkpoint. Using PrivMin or PrivRec on each checkpoint provides differential privacy, but only for an \emph{individual} checkpoint. One must still demonstrate that PrivMin and PrivRec's DP guarantees can tolerate the computations that MinHash Hierarchy conducts using the MinHashes from multiple checkpoints.

Systems that process sensitive information from users must protect their privacy. Current systems based on LSH can provide stronger guarantees if they are enhanced with appropriate differential-privacy mechanisms, but the state of the art in differential privacy is still unable to provide this. Our work demonstrates the need for novel solutions that provide better privacy protections for LSH-based systems.

\section{Related Work}

\textbf{Attacks on FLoC.} Berke et al.~\cite{berke2022privacy} emulated the FLoC algorithm producing cohorts over time, using a proprietary (paid) demographic and browsing history dataset. They then attacked the algorithm using the uniqueness of browsing histories over time and tracking sequences of FLoC IDs. They could identify 95\% of user's devices after 4 weeks. Combining this attack with standard fingerprinting techniques would make it even more effective. In addition, with the observed data, they could connect users' racial backgrounds to their browsing histories, in spite of the fact that they found no direct connection between race and cohorts. Berke et al.'s work focuses on the tracking of individual users and the correlation between cohorts and sensitive demographics. They show that FLoC enables the tracking of individual users, which is an alternative to our Sybil attack. However, they do not reconstruct users' browsing histories as our GAN-IP attack does. Furthermore, our attack also works without the need of collecting data over a long period, which is a requirement for the attack of Berke et al.

Mozilla also mentioned in their report~\cite{floc-mozilla-report} that FLoC was vulnerable to Sybil attacks. However, their claims were neither formally verified nor experimentally validated. Given that FLoC was only tested during a trial with limited user participation, the majority of attacks remained theoretical with no practical implementation. Our work not only gives a theoretical analysis, but also provides a practical implementation and an experimental evaluation using real datasets.

\textbf{Attacks on Perceptual Hashing and NeuralHash.} Another type of LSH is \emph{perceptual hash}, which is used for images. A perceptual hash is 
a fingerprint computed from an input image. It is possible to mount pre-image attacks using conditional adversarial GANs (cGANs), like Pix2Pix~\cite{isola2018imagetoimage}. Such GANs learn how to translate images in one style to another (e.g., translating a hand-drawn sketch of a bag to a photo of a bag). The attack trains a cGAN that learns to translate perceptual hashes to possible pre-images with a success rate of 30\%~\cite{GAN-perceptual-img-hash}. This makes perceptual hash unsuitable for privacy applications. It remains as future work to determine whether cGANs would be successful in mounting pre-image attacks for FLoC or MinHash.

Another instance of a perceptual hashing function is NeuralHash, used by Apple for Child Sexual Abuse Material (\emph{CSAM}) Detection~\cite{apple-csam}. To detect such abusive images, their system stores hashes computed using convolutional neural networks (CNN) and LSH. Their model is vulnerable to adversarial attacks~\cite{neuralhash-attack-github, neuralhash-collider-github} that can lead to non-abusive images being labeled as abusive. To our knowledge, NeuralHash has not been shown to be resistant to pre-image attacks. It remains as future work to investigate what private information the hash reveals about an image.

\textbf{Criticisms to FLoC.} While we are the first who implemented and evaluated FLoC's privacy leakage, we were not the first to criticize it. Other browser vendors pointed out FLoC's potential privacy issues (Mozilla~\cite{floc-mozilla-report}, Brave~\cite{brave-disable-floc}, Vivaldi~\cite{vivaldi-disable-floc}) as well as NGOs (e.g., EFF~\cite{EFF-FLoC-Terrible}) and advertisers~\cite{advertiser-reaction}. For example, FLoC further strengthens already existing hard-to-counter fingerprinting schemes; cohort IDs can be used to further partition users according to browsing behaviors, making tracking easier. Also, FLoC requires a trusted Chrome server that ensures $k$-anonymity and removes sensitive cohorts. However, to the best of our knowledge, a server fulfilling this requirement has not been presented yet. The Chrome server would allow Google to centralize the collection of SimHashes, creating a conflict of interest for Google. This would also strengthen Google's monopoly on advertising and tracking.

\section{Conclusions}

In this work, we studied two systems that use locality sensitive hashing (LSH) to privately handle user data. Both systems considered LSH to be privacy preserving, and, in both cases, we showed how to reconstruct a significant portion of the private inputs from just the hashes. Namely, for MinHash Hierarchy, we extracted parts of vehicle trajectories that were intended to be hidden by the MinHashes computed by the checkpoints. For Google's FLoC, we could construct pre-images to enable an efficient Sybil attack, and from the hashes we reconstructed parts of browsing history. Although Google discontinued FLoC, they had tested it on tens of millions of users underscoring their serious interest in using LSH. Our findings, together with other observed attacks like Apple's Child Sexual Abuse Material Detection, show that the LSH hashes leak substantial information about private data, a fact that is being systematically overlooked.

Our findings show the importance of evaluating the privacy leakage of any system handling sensitive data.  We leave for future work the study of other systems, such as the Topics API, systems that use perceptual hashing, and systems that use differential privacy without a proper evaluation of the information leakage under multiple queries.

\section*{Data Availability}

Further information and updates of this publication are available at \url{https://karelkubicek.github.io/post/floc}. 
Our attacks' implementations are available at \url{https://github.com/privacy-lsh/floc-minhash}. For further details on the FLoC attack, we refer to our technical report~\cite{turati2022analysing}.
The datasets used to evaluate our work are available from the corresponding publications, namely Porto Taxi dataset by Moreira et al.~\cite{moreira2013predicting} and MovieLens by Harper et al.~\cite{Harper2015-cx}.

\begin{acks}
We thank Hung Hoang for his advice on integer programming and Matteo Scarlata for his valuable feedback. We thank the MinHash Hierarchy authors for providing us with the implementation of the MinHash signature computation in their system.
\end{acks}

\bibliographystyle{ACM-Reference-Format}
\bibliography{references}

\appendix

\section{GAN-IP Attack Results}

\begin{table*}[t]
\caption{Example for the GAN-IP attack} \label{fig:gan-output-example}
\centering
\begin{tabular}{lll}
\toprule
\textbf{Target history $h$ from test data:} & \textbf{Generated history $h'$:} & \textbf{Subset of movies selected by int. prog.:} \\ 
\midrule
American President, The (1995) & Ace Ventura: Pet Detective (1994) & Ace Ventura: Pet Detective (1994) \\
Birdcage, The (1996) & Aladdin (1992) & Batman (1989) \\
Client, The (1994) & Batman (1989) & Beauty and the Beast (1991) \\
{\color[HTML]{3531FF} Crimson Tide (1995)} & Beauty and the Beast (1991) & Braveheart (1995) \\
Dances with Wolves (1990) & Braveheart (1995) & Clear and Present Danger (1994) \\
Dead Man Walking (1995) & Clear and Present Danger (1994) & Cliffhanger (1993) \\
{\color[HTML]{3531FF} Die Hard: With a Vengeance (1995)} & Cliffhanger (1993) & {\color[HTML]{3531FF} Crimson Tide (1995)} \\
{\color[HTML]{3531FF} Disclosure (1994)} & {\color[HTML]{3531FF} Crimson Tide (1995)} & {\color[HTML]{3531FF} Disclosure (1994)} \\
English Patient, The (1996) & {\color[HTML]{3531FF} Die Hard: With a Vengeance (1995)} & {\color[HTML]{3531FF} Firm, The (1993)} \\
Fargo (1996) & {\color[HTML]{3531FF} Disclosure (1994)} & Jurassic Park (1993) \\
{\color[HTML]{3531FF} Firm, The (1993)} & {\color[HTML]{3531FF} Firm, The (1993)} & {\color[HTML]{3531FF} Lion King, The (1994)} \\
Forget Paris (1995) & GoldenEye (1995) & {\color[HTML]{3531FF} Outbreak (1995)} \\
Grumpier Old Men (1995) & Jurassic Park (1993) & Pulp Fiction (1994) \\
{\color[HTML]{3531FF} Lion King, The (1994)} & {\color[HTML]{3531FF} Lion King, The (1994)} & {\color[HTML]{3531FF} Seven (a.k.a. Se7en) (1995)} \\
Mirror Has Two Faces, The (1996) & {\color[HTML]{3531FF} Outbreak (1995)} & {\color[HTML]{3531FF} Shawshank Redemption, The (1994)} \\
Mission: Impossible (1996) & Pulp Fiction (1994) & {\color[HTML]{3531FF} Silence of the Lambs, The (1991)} \\
Mrs. Doubtfire (1993) & {\color[HTML]{3531FF} Seven (a.k.a. Se7en) (1995)} & Star Trek: Generations (1994) \\
Mr. Holland's Opus (1995) & {\color[HTML]{3531FF} Shawshank Redemption, The (1994)} & True Lies (1994) \\
Nell (1994) & {\color[HTML]{3531FF} Silence of the Lambs, The (1991)} & {\color[HTML]{3531FF} Twelve Monkeys (1995)} \\
{\color[HTML]{3531FF} Outbreak (1995)} & Star Trek: Generations (1994) & {\color[HTML]{3531FF} Twister (1996)} \\
Philadelphia (1993) & True Lies (1994) & While You Were Sleeping (1995) \\
Postman, The (Postino, Il) (1994) & {\color[HTML]{3531FF} Twelve Monkeys (1995)} &  \\
Rock, The (1996) & {\color[HTML]{3531FF} Twister (1996)} &  \\
Sabrina (1995) & While You Were Sleeping (1995) &  \\
{\color[HTML]{3531FF} Seven (a.k.a. Se7en) (1995)} &  &  \\
{\color[HTML]{3531FF} Shawshank Redemption, The (1994)} &  &  \\
{\color[HTML]{3531FF} Silence of the Lambs, The (1991)} &  &  \\
Spy Hard (1996) &  &  \\
Sudden Death (1995) &  &  \\
Toy Story (1995) &  &  \\
{\color[HTML]{3531FF} Twelve Monkeys (1995)} &  &  \\
{\color[HTML]{3531FF} Twister (1996)} &  & \\
\bottomrule
\end{tabular}
\end{table*}

\end{document}